\documentclass[english, letterpaper, 11pt]{article}

\usepackage{setspace}
\usepackage[T1]{fontenc}
\usepackage[utf8]{inputenc}
\usepackage{babel}
\usepackage{blindtext}
\usepackage{cite}
\usepackage[margin=0.5in]{geometry}
\usepackage[margin=0.5in]{caption}
\usepackage[pdftex]{graphicx}
\graphicspath{{./fig/}}
\usepackage[cmex10]{amsmath}
\usepackage{amssymb}
\usepackage{array}
\usepackage[tight,footnotesize]{subfigure}
\usepackage{stfloats}
\usepackage{algorithm,algorithmic}
\usepackage{enumerate}
\usepackage{amsthm}
\usepackage{color}
\usepackage{mathrsfs}
\usepackage{multicol}
\usepackage{mathtools}
\usepackage{tikz}
\usepackage{pgfplots}
\usepackage{umoline}

\providecommand{\keywords}[1]{\bigskip\textbf{\textit{Index terms---}} #1}

\let\emptyset\varnothing
\newcommand{\norm}[1]{\|{#1}\|}
\newcommand{\diag}{\operatorname{diag}}
\newcommand{\rank}{\operatorname{rank}}

\newcommand{\calF}{\mathcal{F}}
\newcommand{\bbR}{\mathbb{R}}
\newcommand{\bbC}{\mathbb{C}}
\newcommand{\calT}{\mathcal{T}}
\newcommand{\scrT}{\mathscr{T}}
\newcommand{\calR}{\mathcal{R}}
\newcommand{\calC}{\mathcal{C}}
\newcommand{\calX}{\mathcal{X}}
\newcommand{\calY}{\mathcal{Y}}

\newcommand{\calM}{\mathcal{M}}
\newcommand{\rom}[1]{\uppercase\expandafter{\romannumeral #1\relax}}
\newcommand{\transpose}{^{\mathrm{T}}}

\newtheorem{theorem}{Theorem}[section]
\newtheorem{lemma}[theorem]{Lemma}
\newtheorem{corollary}[theorem]{Corollary}
\newtheorem{proposition}[theorem]{Proposition}
\newtheorem{definition}[theorem]{Definition}
\newtheorem{remark}[theorem]{Remark}
\newtheorem{conjecture}[theorem]{Conjecture}

\begin{document}

\title{A Unified Framework for Identifiability Analysis in Bilinear Inverse Problems with Applications to Subspace and Sparsity Models\thanks{This work was supported in part by the National Science Foundation (NSF) under Grants CCF 10-18789 and IIS 14-47879.}}
\author{Yanjun Li\thanks{Department of Electrical and Computer Engineering and Coordinated Science Laboratory, University of Illinois, Urbana-Champaign, IL 61801, USA.}
\and
Kiryung Lee\thanks{Department of Statistics and Coordinated Science Laboratory, University of Illinois, Urbana-Champaign, IL 61801, USA.}
\and
Yoram Bresler\footnotemark[2]}
\date{}
\maketitle

\doublespacing

\abstract
Bilinear inverse problems (BIPs), the resolution of two vectors given their image under a bilinear mapping, arise in many applications. Without further constraints, BIPs are usually ill-posed. In practice, properties of natural signals are exploited to solve BIPs. For example, subspace constraints or sparsity constraints are imposed to reduce the search space. These approaches have shown some success in practice. However, there are few results on uniqueness in BIPs. For most BIPs, the fundamental question of under what condition the problem admits a unique solution, is yet to be answered. For example, blind gain and phase calibration (BGPC) is a structured bilinear inverse problem, which arises in many applications, including inverse rendering in computational relighting (albedo estimation with unknown lighting), blind phase and gain calibration in sensor array processing, and multichannel blind deconvolution (MBD). It is interesting to study the uniqueness of such problems.

In this paper, we define identifiability of a BIP up to a group of transformations. We derive necessary and sufficient conditions for such identifiability, i.e., the conditions under which the solutions can be uniquely determined up to the transformation group. Applying these results to BGPC, we derive sufficient conditions for unique recovery under several scenarios, including subspace, joint sparsity, and sparsity models. For BGPC with joint sparsity or sparsity constraints, we develop a procedure to compute the relevant transformation groups. We also give necessary conditions in the form of tight lower bounds on sample complexities, and demonstrate the tightness of these bounds by numerical experiments. The results for BGPC not only demonstrate the application of the proposed general framework for identifiability analysis, but are also of interest in their own right.

\keywords{uniqueness, transformation group, equivalence class, ambiguity, blind gain and phase calibration, sensor array processing, inverse rendering, SAR autofocus, multichannel blind deconvolution}


\section{Introduction}
Whereas linear inverse problems are well-understood and the literature on them is vast, much less is known about  bilinear inverse problems (BIPs). BIPs, i.e., recovering two variables $x$ and $y$ given a bilinear measurement $z=\calF(x,y)$, have attracted considerable attention recently. However, in spite of recent progress, the question of identifiability - or uniqueness of the solutions in BIPs under a variety of realistic conditions, has been largely open. BIPs arise in many important applications, such as blind deconvolution \cite{Kundur1996,Abed-meraim1997}, phase retrieval \cite{Taylor1981,Fienup1982}, dictionary learning \cite{Rubinstein2010}, etc. These problems usually involve recovering the inputs of an under-determined bilinear system. They also suffer from scaling ambiguity among other possible ambiguities (e.g., shift ambiguity of blind deconvolution, multiplication by a permutation matrix in dictionary learning, multiplication by an arbitrary invertible matrix in matrix factorization problems, etc). Therefore, these problems are ill-posed and do not yield unique solutions. By introducing further constraints that exploit the properties of natural signals, one can reduce the search space, which may help identifiability. For example, cone constraints, such as convex cone constraints (e.g., positivity), subspace constraints, and union of subspaces constraints (e.g., sparsity or joint sparsity), are very common in BIPs. However, even with a reduced feasible set, a BIP often still exhibits some ambiguities, such as scaling \cite{Choudhary2013a}. 

In this paper, to address the issues of ambiguity, we expand the notion of identifiability of BIPs. We resolve the ambiguity issues by allowing uniqueness up to a group of transformations, which define equivalence classes of solutions. We then derive necessary and sufficient conditions for identifiability in BIPs up to the transformation group. We demonstrate the utility of this  proposed framework and the new general results by applying them to a prototypical BIP - the blind gain and phase calibration (BGPC) problem, with subspace, or with sparsity constraints. The results include the first algebraic sample complexity conditions for this problem. Owing to the ubiquitousness of BGPC, the results obtained here for this problem are of broad interest in their own right.

Blind gain and phase calibration (BGPC) is a bilinear inverse problem that arises in many applications. It is the joint recovery of an unknown gain and phase vector $\lambda$ and signal vectors $\phi_1,\phi_2,\cdots,\phi_N$ given the entrywise product $Y=\diag(\lambda)\Phi$, where $\Phi=[\phi_1,\phi_2,\cdots,\phi_N]$. In inverse rendering \cite{Nguyen2013}, when the surface profile (3D model) of the object is known, the joint recovery of the albedo\footnote{Albedo, also known as reflection coefficient, is the ratio of reflected radiation from a surface to incident radiation upon it.} and the lighting conditions is a BGPC problem. In sensor array processing \cite{Paulraj1985}, if the directions of arrival of source signals are properly discretized using a grid, and the sensors have unknown gain and phase, the joint recovery of the source signals and the gain and phase of the sensors is a BGPC problem. In multichannel blind deconvolution (MBD) with the circular convolution model, the joint recovery of the signal and multiple channels is a BGPC problem. In all these problems, it is common to impose subspace, joint sparsity, or sparsity constraints on the signals represented by the columns of $\Phi$.

In this paper, after deriving general necessary and sufficient conditions for identifiability in a BIP up to the transformation group, we apply these to BGPC and give identifiability results under several scenarios. We first consider a subspace constraint and provide an alternative proof for the result in inverse rendering \cite{Nguyen2013}. Then we consider a joint sparsity constraint. We develop a procedure to determine the relevant equivalence classes and transformation groups for different bases. Then we give sufficient conditions for the identifiability of jointly sparse signals (1D or 2D), or piecewise constant signals. For BGPC with subspace or joint sparsity constraints, we also give necessary conditions in the form of tight lower bounds on sample complexities. We show that the sufficient conditions and the necessary conditions coincide in some cases. We design algorithms to check the identifiability of given signals and demonstrate the tightness of our sample complexity bounds. We analyze the gaps and present conjectures about how to bridge them. Then we derive a universal sufficient condition for BGPC with a sparsity constraint. This condtion is the most stringent, but applies to all bases and all equivalence classes of solutions. Once the condition is met, the solution of the BGPC problem can be recovered uniquely up to an unknown generalized permutation, regardless of the basis.

The rest of the paper is organized as follows. We introduce the problem setup and compare our results with related work in the rest of this section. In Section \ref{sec:bip}, we formally define ``transformation group'' and ``identifiability'' for bilinear inverse problems and derive sufficent and necessary conditions for identifiability up to a transformation group. We derive sufficient conditions for the identifiability of BGPC with a subspace constraint, a joint sparsity constraint, and a sparsity constraint, in Sections \ref{sec:subspace}, \ref{sec:jointsparsity}, and \ref{sec:sparsity} respectively. We also give necessary conditions and analyze the tightness of our sample complexity bounds in Sections \ref{sec:subspace} and \ref{sec:jointsparsity}.
We conclude this paper in Section \ref{sec:conclusions} with some discussion and open problems.

\subsection{Notations}
Before proceeding to the problem statement, we state the notations that will be used throughout the paper. We use upper-case letters $A$, $X$ and $Y$ to denote matrices, and lower-case letters to denote vectors. The diagonal matrix whose diagonal entries are the entries of vector $\lambda$ is denoted by $\diag(\lambda)$. We use $I$ to denote the identity matrix and $F$ to denote the normalized discrete Fourier transform (DFT) matrix. Unless otherwise stated, all vectors are column vectors. The dimensions of all vectors and matrices are made clear in the context. A vector is said to be non-vanishing, if all its entries are nonzero.

We use $\Omega_\calX,\Omega_\calY$ to denote subsets of vector spaces $\calX,\calY$. The Cartesian product of two sets is denoted by $\Omega_\calX\times \Omega_\calY$. An element of $\Omega_\calX\times \Omega_\calY$ is denoted by $(x,y)$, where $x\in \Omega_\calX$ and $y\in \Omega_\calY$. We use $\scrT_\calX$ and $\scrT_\calY$ to denote transformation groups (to be defined in Section \ref{sec:tgec}). The Cartesian product of two transformation groups $\scrT_\calX,\scrT_\calY$ (also known as direct product in group theory terminology) is denoted by $\scrT_\calX\times \scrT_\calY$. Elements of the transformation groups are denoted by $\calT_\calX\in \scrT_\calX$, $\calT_\calY\in \scrT_\calY$ and $(\calT_\calX,\calT_\calY)\in \scrT_\calX\times \scrT_\calY$.

We use $j,k$ to denote indices, and $J,K$ to denote index sets. If a matrix or a vector has dimension $n$, then an index set $J$ is a subset of $\{1,2,\cdots,n\}$. We use $|J|$ to denote the cardinality of $J$, and $J^c$ to denote its complement. We use superscript letters to denote subvectors or submatrices. Thus, $x^{(J)}$ represents the subvector of $x$ consisting of the entries indexed by $J$. The scalar $x^{(j)}$ represents the $j$th entry of $x$. The submatrix $A^{(J,K)}$ has size $|J|\times |K|$ and consists of the entries indexed by $J\times K$. The vector $A^{(:,k)}$ represents the $k$th column of the matrix $A$. The colon notation is inherited from MATLAB.

We use $./$ and $\odot$ to denote entrywise division and entrywise product, respectively. Circular convolution is denoted by $\circledast$. The direct sum of two subspaces is denoted by $\oplus$. The kronecker product of two matrices is denoted by $\otimes$. The row space and column space of a matrix are denoted by $\calR(\cdot)$ and $\calC(\cdot)$, respectively.

\subsection{Problem Statement}
We formally state the general bilinear inverse problem (BIP) and a special BIP termed blind gain and phase calibration (BGPC) in this section. First, a bilinear mapping is defined as follows.

\begin{definition}
Let $\mathcal{X}$, $\mathcal{Y}$ and $\mathcal{Z}$ be three linear vector spaces.
A bilinear mapping is a function $\calF: \mathcal{X}\times \mathcal{Y}\rightarrow \mathcal{Z}$ such that for any $y\in \mathcal{Y}$ the mapping $x \mapsto \calF(x,y)$ is a linear mapping from $\mathcal{X}$ to $\mathcal{Z}$ and for any $x\in \mathcal{X}$ the mapping $y \mapsto \calF(x,y)$ is a linear mapping from $\mathcal{Y}$ to $\mathcal{Z}$ .
\end{definition}

Given the measurement $z = \calF(x_0,y_0)$, the following feasibility problem is called the unconstrained bilinear inverse problem:
\begin{align*}
\text{(Unconstrained BIP)}\qquad \text{find}~~&(x,y)\in \mathcal{X}\times \mathcal{Y},\\
\text{s.t.}~~&\calF(x,y) = z.
\end{align*}

Bilinear inverse problems are usually underdetermined, and hence do not yield unique solutions. A variety of constraints $x\in \Omega_\calX\subset \mathcal{X}, y\in \Omega_\calY\subset \mathcal{Y}$ can be imposed to reduce the search space and make the problem better-posed. The constrained bilinear inverse problem is:
\begin{align}\label{eq:19}
\begin{split}
\text{(Constrained BIP)}\qquad\text{find}~~&(x,y),\\
\text{s.t.}~~&\calF(x,y) = z,\\
& x \in \Omega_\calX,~y \in \Omega_\calY.
\end{split}
\end{align}

\emph{Blind gain and phase calibration (BGPC)} is the following constrained BIP given the measurement $Y=\diag(\lambda_0)\Phi_0$:
\begin{align*}
\text{find}~~&(\lambda,\Phi),\\
\text{s.t.}~~&\diag(\lambda)\Phi = Y,\\
& \lambda \in \Omega_\Lambda,~\Phi \in \Omega_\Phi.
\end{align*}
where $\lambda \in \Omega_\Lambda\subset\bbC^{n}$ is the unknown gain and phase vector, $\Phi\in\Omega_\Phi\subset\bbC^{n\times N}$ is the signal matrix. In this paper, we impose no constraints on $\lambda$, i.e., $\Omega_\Lambda = \bbC^n$.
As for the matrix $\Phi$, we impose subspace, joint sparsity, or sparsity constraints. In all three scenarios, $\Phi$ can be represented in the factorized form $\Phi=AX$, where the columns of $A\in \bbC^{n\times m}$ form a basis or a frame (an overcomplete dictionary), and $X\in \Omega_\calX\subset \bbC^{m\times N}$ is the matrix of coordinates. The constraint set becomes $\Omega_\Phi=\{\Phi=AX:X\in\Omega_\calX\}$. Under some mild conditions\footnote{Under a subspace constraint, $A$ is required to have full column rank. Under a joint sparsity or sparsity constraint, $A$ is required to satisfy the spark condition \cite{Donoho2003}.} on $A$, the uniqueness of $\Phi$ is equivalent to the uniqueness of $X$. For simplicity, we treat the following problem as the BGPC problem from now on.
\begin{align*}
\text{(BGPC)}\qquad\text{find}~~&(\lambda,X),\\
\text{s.t.}~~&\diag(\lambda)AX = Y,\\
& \lambda \in \bbC^n,~X \in \Omega_\calX.
\end{align*}

Next, we elaborate on the three scenarios considered in this paper:

(\rom{1}) \emph{Subspace constraints.} The signals represented by the columns of $\Phi$ reside in a low-dimensional subspace spanned by the columns of $A$. The matrix $A$ is tall ($n>m$) and has full column rank. The constraint set is $\Omega_\calX=\bbC^{m\times N}$.
	
	In inverse rendering \cite{Nguyen2013}, the columns of $Y=\diag(\lambda)\Phi$ represent images under different lighting conditions, where $\lambda$ represents the unknown albedos\footnote{In inverse rendering, albedos are real and positive. We ignore this extra information here for simplicity.}, and the columns of $\Phi$ represent the intensity maps of incident light. The columns of $A$ are the first several spherical harmonics extracted from the 3D model of the object. They form a basis of the low-dimensional subspace in which the intensity maps reside.
	
	Multichannel blind deconvolution (MBD) with the circular convolution model also falls into this category. The measurement $Y^{(:,j)}=\diag(\lambda)\Phi^{(:,j)}$ can be also written as:	
	\[
F^*Y^{(:,j)}=\frac{1}{\sqrt{n}}(F^*\lambda)\circledast(F^*\Phi^{(:,j)}).
	\]
	The vector $\lambda$ represents the DFT of the signal, and columns of $\Phi$ represent the DFT of the channels. The columns of $F^*A$ form a basis for the low-dimensional subspace in which the channels reside. For example, when the multiple channels are FIR filters that share the same support $J$, they reside in a low-dimensional subspace whose basis is $F^*A=I^{(:,J)}$. By symmetry, the roles of signals and channels can be switched.
In channel encoding, when multiple signals are encoded by the same tall matrix $E$, they reside in a low-dimensional subspace whose basis is $F^*A=E$. In this case, the vector $\lambda$ represents the DFT of the channel.
	
(\rom{2}) \emph{Joint sparsity constraints.} The columns of $\Phi$ are jointly sparse over a dictionary $A$, where $A$ is a square matrix ($n=m$) or a fat matrix ($n<m$). 
	The constraint set $\Omega_\calX$ is \[\Omega_\calX=\{X\in\bbC^{m\times N}: \text{$X$ has at most $s$ nonzero rows}\}.\]
In other words, the columns of $X$ are jointly $s$-sparse.
	
	In sensor array processing with uncalibrated sensors, the vector $\lambda$ represents unknown gain and phase for the sensors, and the columns of $\Phi$ represent array snapshots captured at different time instants. If the direction of arrival (DOA) is discretized using a grid, then each column of $A$ represents the array response of one direction on the grid. With only $s$ unknown sources, each column of $\Phi$ is the superposition of the same $s$ columns of $A$. Hence the columns of the source matrix $X$ are jointly $s$-sparse. 
	
	In synthetic aperture radar (SAR) autofocus \cite{Morrison2009}, which is a special multichannel blind deconvolution problem, $X$ represents the SAR image and $A=F$ is the 1D DFT matrix. The entries in $\lambda$ represents the phase error in the Fourier imaging data, which varies only along the cross-range dimension\footnote{In SAR autofocus, the entries of the phase error $\lambda$ have unit moduli. We ignore this extra information here for simplicity.}. If we extend the coverage of the image by oversampling the Fourier domain in the cross-range dimension, the rows of the image $X$ corresponding to the region that is not illuminated by the antenna beam are zeros. Thus, the SAR image $X$ can be modeled as a matrix with jointly sparse columns.
	
(\rom{3}) \emph{Sparsity constraints.} The matrix $\Phi$ is sparse over a dictionary $A$, where $A$ is a square matrix ($n=m$) or a fat matrix ($n<m$). The constraint set $\Omega_\calX$ is
\[\Omega_\calX=\{X\in\bbC^{m\times N}:\text{$X$ has at most $s$ nonzero entries}\}.\]
A matrix $X$ with sparse columns can be considered as a special case of this scenario.
	
	Consider the following multichannel blind deconvolution problem. An acoustic signal is transmitted under reverberant conditions and recorded by a microphone array. The DFT of the signal is $\lambda$, $A=F$ is the DFT matrix, each column of $\Phi=AX$ is the DFT of the channel of a corresponding microphone, and the corresponding column of $X$ is a sparse multipath channel that contains nonzero values at a few locations. 

In the rest of this paper, we will address the identifiability of the general BIP and the above BGPC problem. For BGPC, the constraint sets $\Omega_\Lambda$ and $\Omega_\calX$ are closed under scalar multiplication.
For any nonzero scalar $\sigma$, the pairs $(\lambda_0,X_0)$ and $(\sigma\lambda_0,\frac{1}{\sigma}X_0)$ map to the same $Y$ and hence are non-distinguishable. We say that this problem suffers from scaling ambiguity. The set $\{(\sigma\lambda_0,\frac{1}{\sigma}X_0):\lambda\in \Omega_\Lambda, X\in \Omega_\calX, \sigma\neq 0\}$ is an equivalence class of solutions generated by a group of scaling transformations. More complex ambiguities and equivalence classes will be analyzed later. Our identifiability results answer the question under what conditions the solution $(\lambda_0,X_0)$ is unique up to scaling, or up to other transformation groups.

\subsection{Related Work}
Recently, solving bilinear or quadratic inverse problems with the methodology of ``lifting'' has attracted much attention. Examples include recent works on blind deconvolution \cite{Ahmed2014} and phase retrieval \cite{Candes2013a,Candes2013,Candes2013b}. In the lifting framework, for any bilinear mapping $\calF:\bbC^m\times\bbC^n\rightarrow\mathcal{Z}$, there exists a linear operator $\mathcal{G}:\bbC^{m\times n}\rightarrow \mathcal{Z}$ such that $\mathcal{G}(xy\transpose)=\calF(x,y)$. Given the measurement $z=\mathcal{G}(x_0y_0\transpose)=\calF(x_0,y_0)$, one can recast the BIP as the recovery of the rank-$1$ matrix $x_0y_0\transpose \in \Omega_\calM=\{xy\transpose:x\in \Omega_\calX,y\in \Omega_\calY\}$.
\begin{align*}
\text{(Lifted BIP)}\qquad\text{find}~~& M,\\
\text{s.t.}~~&\mathcal{G}(M) = z,\\
& M\in \Omega_\calM.
\end{align*} 
Choudhary and Mitra \cite{Choudhary2013a} adopted this framework, and showed that the lifted BIP has a unique solution $M_0=x_0y_0\transpose$ if the null space of $\mathcal{G}$ does not contain the difference of $M_0$ and any other matrix in $\Omega_\calM$, i.e.,
\[
\mathcal{N}(\mathcal{G})\bigcap \{M_0-M:M\in \Omega_\calM\} = \{0\}.
\]
The identifiability analysis hinges on finding the set of rank-$2$ matrices in the null space of $\mathcal{G}$. They addressed the question of identifiability in an abstract BIP under the assumptions that the set of rank-$2$ matrices in $\mathcal{N}(\mathcal{G})$ has low complexity (e.g., finite cardinality or small covering number). Using this framework, they showed that blind deconvolution with a canonical sparsity prior is \emph{not} identifiable \cite{Choudhary2014a}.

In contrast, we create a more general framework for the identifiability of BIPs. We consider bilinear mappings defined on general vector spaces (not just Euclidean spaces). Besides scaling ambiguity, our framework allows other ambiguities. We extend the notion of identifiability to identifiability up to transformation groups. Our framework is amenable to BIPs with matrix multiplications, such as dictionary learning \cite{Spielman2013,Agarwal2013,Agarwal2013a,Arora2013,Arora2014} and the BGPC problem. For the BGPC problem, we are able to derive identifiability results under subspace, joint sparsity, or sparsity constraints within our framework. Furthermore, we provide an explicit enumeration-based scheme to determine, under subspace or joint sparsity constraints, the identifiability of a solution for given measurements. No such results are available within the lifting framework: under the same constraints, it is not obvious how to find the set of rank-$2$ matrices in $\mathcal{N}(\mathcal{G})$ within the lifting framework. Moreover, since the set of rank-$2$ matrices in $\mathcal{N}(\mathcal{G})$ may be infinite, it is also unclear how to check the identifiability condition for any given scenario.

Other related work has to do with instances of the BGPC problem. The structure of the BGPC problem arises in many signal processing applications. In each of these, the problem formulation and treatment were tailored to the application. Instead, we address the identifiability of all these problems within the one common framework. Nguyen et al. \cite{Nguyen2013} showed a sufficient condition for unique inverse rendering, which falls into the category of BGPC problems with subspace constraints. By examining the problem in our framework, we are able to replicate Nguyen's result and provide an alternative proof. In addition, we give a new necessary condition that features a tight lower bound. Morrison et al. \cite{Morrison2009} proposed an algorithm for SAR autofocus and showed a necessary condition for their algorithm. If the support is unknown, the SAR autofocus problem falls into the category of BGPC problems with joint sparsity constraints. Using our notion of identifiability up to a transformation group, we provide a sufficient condition for unique recovery up to an unknown scaling and a circular shift. Most works on the identifiability of MBD considered the linear convolution model \cite{Moulines1995,Abed-meraim1997}. These traditional works used finite impulse response (FIR) models, and never incorporated joint sparsity, or sparsity. In contrast, we consider the circular convolution model, which is more challenging in that the circular convolution with a vector can be non-injective, while the linear convolution with a vector is always injective. On the other hand, the circular convolution model is more general. By zero padding the signal and the channels (equivalent to Fourier domain oversampling), linear convolutions can be rewritten as circular convolutions with a support constraint. That falls into the category of BGPC with a subspace constraint. As an important extension of the theory of MBD, we study in this paper MBD with subspace, joint-sparsity, and sparsity constraints.


\section{Identifiability of Bilinear Inverse Problems}\label{sec:bip}
\subsection{Transformation Groups and Equivalence Classes}\label{sec:tgec}
An important question concerning a bilinear inverse problem is to determine when it admits a unique solution. To formulate a good answer, we need to be able to handle the ambiguities of a bilinear inverse problem. For any nonzero scalar $\sigma$ such that $\sigma x_0 \in \Omega_\calX$ and $\frac{1}{\sigma}y_0 \in \Omega_\calY$, by bilinearity, $\calF(\sigma x_0,\frac{1}{\sigma}y_0)=\calF(x_0,y_0)=z$. Therefore, the constrained BIP does not yield a unique solution if $\Omega_\calX,\Omega_\calY$ contain such scaled versions of $x_0,y_0$. That is called scaling ambiguity. When $\Omega_\calX,\Omega_\calY$ are closed under scalar multiplication (e.g., subspaces or unions of subspaces), the set $[(x_0,y_0)]=\{(\sigma x_0,\frac{1}{\sigma} y_0):x_0\in \Omega_\calX, y_0\in \Omega_\calY, \sigma\neq 0\}$ is an equivalence class with an exemplar $(x_0,y_0)$. The transformation $\calT: \Omega_\calX\times \Omega_\calY\rightarrow \Omega_\calX\times \Omega_\calY$ such that $\calT(x,y)=(\sigma x,\frac{1}{\sigma}y)$ is an equivalence transformation. The set of all such transformations
\begin{equation}
\scrT = \{\calT: \calT(x,y)=(\sigma x,\frac{1}{\sigma}y),\text{ for some nonzero $\sigma\in\bbC$}\}
\label{eq:24}
\end{equation}
forms a transformation group. In group theory terminology, the equivalence class $[(x_0,y_0)]$ is the orbit of $(x_0,y_0)$ under the action of $\scrT$ \cite{Bredon1972}. Any valid definition of unique recovery must include uniqueness up to scaling, i.e., the equivalence class $[(x_0,y_0)]$ can be uniquely identified.
There can be other ambiguities for a particular bilinear inverse problem (e.g., shift ambiguity of blind deconvolution). We need formal definitions of transformation groups and equivalence classes before proceeding towards identifiability.

\begin{definition}\label{def:tg}
A set $\scrT_\calX$ of transformations from $\Omega_\calX$ to itself is said to be a transformation group on $\Omega_\calX$, if the following properties hold:
\begin{enumerate}
	\item For any $\calT_{\calX,1},\calT_{\calX,2}\in\scrT_\calX$, the composition of the two transformations $\calT_{\calX,2}\circ\calT_{\calX,1}$ belongs to $\scrT_\calX$.
	\item $\scrT_\calX$ contains identity transformation $\mathbf{1}_\calX(x)=x$ for all $x\in \Omega_\calX$.
	\item For any $\calT_\calX \in\scrT_\calX$, there exists $\calT_\calX^{-1}\in\scrT_\calX$ such that $\calT_\calX^{-1}\circ\calT_\calX=\calT_\calX\circ\calT_\calX^{-1}=\mathbf{1}_\calX$.
\end{enumerate}
\end{definition}

If $\scrT_\calX,\scrT_\calY$ are transformation groups on $\Omega_\calX,\Omega_\calY$ respectively, then their direct product $\scrT_\calX\times \scrT_\calY$ is a transformation group on $\Omega_\calX\times \Omega_\calY$. The action of $(\calT_\calX,\calT_\calY)\in \scrT_\calX\times \scrT_\calY$ on $(x,y)\in \Omega_\calX\times \Omega_\calY$ is $(\calT_\calX(x),\calT_\calY(y))$. If there exists $\calT=(\calT_\calX,\calT_\calY)\in \scrT_\calX\times \scrT_\calY$, such that
\[
\calF(\calT(x,y))=\calF(\calT_\calX(x),\calT_\calY(y))=\calF(x,y),
\]
for all $(x,y)\in \Omega_\calX\times \Omega_\calY$, then $\calT$ maps a pair $(x,y)$ to another pair $(\calT_\calX(x),\calT_\calY(y))$ so that the two pairs cannot be distinguished by their images under $\calF$. If a set of such $\calT$'s form a subgroup of $\scrT_\calX\times\scrT_\calY$, we have a \emph{transformation group associated with the bilinear mapping} $\calF$.

\begin{definition}\label{def:tgb}
A transformation group $\scrT$ on $\Omega_\calX\times \Omega_\calY$ is said to be a transformation group associated with the bilinear mapping $\calF$ if:
\begin{enumerate}
	\item $\scrT\subset\scrT_\calX\times \scrT_\calY$ is a subgroup of the direct product of two transformation groups $\scrT_\calX$ and $\scrT_\calY$, on $\Omega_\calX$ and $\Omega_\calY$, respectively.
	\item For all $(x,y)\in \Omega_\calX\times \Omega_\calY$ and for all $\calT\in \scrT$, $\calF(x,y) = \calF(\calT(x,y))$. Or equivalently, $\calF=\calF\circ\calT$ for all $\calT\in\scrT$.
\end{enumerate}
\end{definition}

To enable an identifiability result up to a transformation group (see Section \ref{sec:iutg}), the transformation group must capture \emph{all} inherent ambiguities of the BIP. This motivates the following definition of the \emph{ambiguity transformation group of the bilinear mapping}.
\begin{definition}\label{def:atgb}
A transformation group $\scrT$ on $\Omega_\calX\times \Omega_\calY$ is said to be the ambiguity transformation group of the bilinear mapping $\calF$ if $\scrT$ is the largest transformation group associated with $\calF$, i.e., if $\scrT$ contains all transformation groups associated with $\calF$.
A transformation $\calT$ in the ambiguity transformation group $\scrT$ of the bilinear mapping $\calF$ is said to be an equivalence transformation associated with $\calF$.
\end{definition}

Next, we define an \emph{equivalence class associated with the bilinear inverse problem}.
\begin{definition}\label{def:ec}
Given the ambiguity transformation group $\scrT$ of the bilinear mapping $\calF$ on $\Omega_\calX\times \Omega_\calY$, and $(x_0,y_0)\in\Omega_\calX\times\Omega_\calY$, the set 
\[
[(x_0,y_0)]_\scrT = \{(x,y)\in \Omega_\calX\times \Omega_\calY: (x,y)=\calT(x_0,y_0)\text{ for some }\calT\in\scrT\}
\]
is called the equivalence class of $(x_0,y_0)$ associated with the bilinear inverse problem in \eqref{eq:19}. In group theory terminology, $[(x_0,y_0)]_\scrT$ is called the orbit of $(x_0,y_0)$ under the action of $\scrT$.
\end{definition}

\begin{definition}
Given the ambiguity transformation group $\scrT$ of the bilinear mapping $\calF$ on $\Omega_\calX\times \Omega_\calY$, and $x_0\in\Omega_\calX$, the set
\[
[x_0]^L_{\scrT}=\{x\in\Omega_\calX:\exists y_0,y\in\Omega_\calY, \text{ s.t. } (x,y)\in[(x_0,y_0)]_{\scrT}\}
\]
is called the left equivalence class of $x_0$. 

Similarly, given the ambiguity transformation group $\scrT$ of the bilinear mapping $\calF$ on $\Omega_\calX\times \Omega_\calY$, and $y_0\in\Omega_\calY$, the set
\[
[y_0]^R_{\scrT}=\{y\in\Omega_\calY:\exists x_0,x\in\Omega_\calX, \text{ s.t. } (x,y)\in[(x_0,y_0)]_{\scrT}\}
\]
is called the right equivalence class of $y_0$.
\end{definition}

The definition of a transformation group guarantees that the relation between elements in an orbit satisfies reflexivity, transitivity and symmetry. Therefore, an orbit is an equivalence class. If $\scrT$ is the ambiguity transformation group of the bilinear mapping $\calF$, then all the elements in the equivalence class $[(x_0,y_0)]_\scrT$ share the same image under $\calF$. Therefore, they are equivalent solutions to the bilinear inverse problem in \eqref{eq:19}. In fact, under some mild conditions on the bilinear mapping, Definitions \ref{def:tgb} and \ref{def:atgb} have additional implications.

\begin{proposition}\label{pro:linearity}
Assume that the bilinear mapping $\calF$ has no non-trivial left annihilator of $\Omega_\calY$, i.e., if $\calF(x_0,y)=0$ for all $y\in \Omega_\calY$, then $x_0=0$.
Then every equivalence transformation $\calT=(\calT_\calX,\calT_\calY)\in\scrT$ satisfies the following:
\begin{itemize}
	\item If $0\in \Omega_\calX$, then $\calT_\calX(0)=0$.
	\item For $x_1,x_2\in \Omega_\calX$ and scalars $a_1,a_2$, if $a_1x_1+a_2x_2\in \Omega_\calX$, then \[\calT_\calX(a_1x_1+a_2x_2)=a_1\calT_\calX(x_1)+a_2\calT_\calX(x_2).\] If $\Omega_\calX$ is a linear vector space, then $\calT_\calX$ is a linear transformation.
\end{itemize}

Similarly, assume that the bilinear mapping $\calF$ has no non-trivial right annihilator of $\Omega_\calX$, i.e., if $\calF(x,y_0)=0$ for all $x\in \Omega_\calX$, then $y_0=0$. Then every equivalence transformation $\calT=(\calT_\calX,\calT_\calY)\in\scrT$ satisfies the following:
\begin{itemize}
	\item  If $0\in \Omega_\calY$, then $\calT_\calY(0)=0$.
	\item For $y_1,y_2\in \Omega_\calY$ and scalars $b_1,b_2$, if $b_1y_1+b_2y_2\in \Omega_\calY$, then \[\calT_\calY(b_1y_1+b_2y_2)=b_1\calT_\calY(y_1)+b_2\calT_\calY(y_2).\] If $\Omega_\calY$ is a linear vector space, then $\calT_\calY$ is a linear transformation.
\end{itemize}
\end{proposition}

\begin{proof}
Due to the symmetry, we only need to prove the results for $\calT_\calX$.

If $0\in \Omega_\calX$, then $\calF(\calT_\calX(0),y) = \calF(\calT(0,\calT_\calY^{-1}(y)))=\calF(0,\calT_\calY^{-1}(y))=0$ for all $y\in \Omega_\calY$. By assumption, there is no non-trivial left annihilator of $\Omega_\calY$. Therefore, $\calT_\calX(0)=0$.

If $a_1x_1+a_2x_2\in \Omega_\calX$, then
\begin{align*}
& \calF(\calT_\calX(a_1x_1+a_2x_2),y)\\
= & \calF(\calT(a_1x_1+a_2x_2,\calT_\calY^{-1}(y)))\\
= & \calF(a_1x_1+a_2x_2,\calT_\calY^{-1}(y))\\
= & a_1\calF(x_1,\calT_\calY^{-1}(y))+a_2\calF(x_2,\calT_\calY^{-1}(y))\\
= & a_1\calF(\calT_\calX(x_1),y)+a_2\calF(\calT_\calX(x_2),y)\\
= & \calF(a_1\calT_\calX(x_1)+a_2\calT_\calX(x_2),y).
\end{align*}
Then $\calF(\calT_\calX(a_1x_1+a_2x_2)-(a_1\calT_\calX(x_1)+a_2\calT_\calX(x_2)),y)=0$ for all $y\in \Omega_\calY$. There is no non-trivial left annihilator of $\Omega_\calY$. Hence $\calT_\calX(a_1x_1+a_2x_2)=a_1\calT_\calX(x_1)+a_2\calT_\calX(x_2)$, and $\calT_\calX$ is a linear transformation if $\Omega_\calX$ is a linear vector space.
\end{proof}

Bilinear mappings that arise in applications usually have no non-trivial left or right annihilators. Therefore, common equivalence transformations, such as scaling and shift, are linear transformations. However, there are examples where equivalence transformations are nonlinear (cf. Appendix \ref{app:nta}).

Before proceeding to identifiability, let us consider the following blind deconvolution problem as a concrete example. The measurement is $z=x_0\circledast y_0\in \bbC^{n}$.
\begin{align*}
\text{find}~~&(x,y),\\
\text{s.t.}~~&x\circledast y = z,\\
& x \in \bbC^n,~y \in \bbC^n.
\end{align*}
Define transformation groups $\scrT_\calX,\scrT_\calY$ on $\calX=\calY=\bbC^n$:
\[
\scrT_\calX=\scrT_\calY=\{\calT_{\bbC^n}: \calT_{\bbC^n}(x)=\sigma S_\ell (x), \text{ for some nonzero $\sigma\in \bbC$ and some integer $\ell$}\},
\]
where the linear transformation $S_\ell$ is the circular shift by $\ell$, defined as follows. If $x=S_\ell(x_0)$, then $x^{(j)}=x_0^{(k)}$ for all $1\leq j,k\leq n$ where $j-k=\ell$ (modulo $n$). Then the following subgroup $\scrT\subset\scrT_\calX\times \scrT_\calY$ is a transformation group associated with circular convolution:
\begin{equation}
\scrT = \left\{\calT: \calT(x,y)=\left(\sigma S_\ell (x),\frac{1}{\sigma} S_{-\ell} (y)\right),  \text{ for some nonzero $\sigma\in \bbC$ and some integer $\ell$}\right\}.
\end{equation}
Note that $\scrT$ is a transformation group associated with circular convolution, and a subgroup of $\scrT_\calX\times\scrT_\calY$. However, it is not separable, i.e., it cannot be written as the direct product of two transformation groups. Furthermore, $\scrT$ is not the ambiguity transformation group, because it does not capture all the ambiguities of the above blind deconvolution problem. For example, there exist non-trivial vectors $u,v\in\bbC^n$ such that $u\circledast v$ is the kronecker delta. Thus, $(x\circledast u, y\circledast v)$ is an equivalent pair of $(x,y)$. The set of such transformations is not contained in $\scrT$.

\subsection{Identifiability up to a Transformation Group}\label{sec:iutg}

The concept of identifiability should be generalized to allow unique recovery up to the ambiguity transformation group. If the equivalence class containing the solution can be uniquely identified, the solution is considered identifiable.
\begin{definition}\label{def:iutg}
In the constrained BIP, the solution $(x_0,y_0)$ in which $x_0\neq 0,y_0\neq 0$ is said to be identifiable up to a transformation group $\scrT$, if every solution $(x,y)$ satisfies that $(x,y)=\calT(x_0,y_0)$ for some $\calT\in\scrT$, or equivalently, $(x,y)\in[(x_0,y_0)]_{\scrT}$.
\end{definition}

In general, the ambiguity transformation group for a certain BIP may not be known a priori. It may require some insight to capture all the ambiguities inherent in the problem. However, we can tell whether or not a given transformation group is the ambiguity transformation group by checking the identifiability. If there exists an identifiability result up to this transformation group, it has to be the largest.
If the constraint sets $\Omega_\calX$ and $\Omega_\calY$ are closed under scalar multiplication, then one can start by checking the group of scaling transformations defined in \eqref{eq:24}. For some BIPs, the ambiguities go beyond scaling ambiguity. Hence we have to choose larger transformation groups. An example is BGPC with a joint sparsity constraint (Section \ref{sec:equivalence}).

We derive a necessary and sufficient condition for identifiability in Theorem \ref{thm:ibip}, and a more intuitive sufficient condition in Corollary \ref{cor:ibip}. Here is how we interpret these results: In order to prove that certain conditions are sufficient to guarantee identifiability up to a transformation group, it suffices to first show that $x_0$ can be identified up to the transformation group; and then show that once $x_0$ is identified and substituted in the problem, $y_0$ can be identified.
By the symmetry of the problem, we can derive another sufficient condition by switching the roles of $x_0$ and $y_0$.

\begin{theorem}\label{thm:ibip}
In the constrained BIP, the pair $(x_0,y_0)$ ($x_0\neq 0,y_0\neq 0$) is identifiable up to $\scrT$ if and only if the following two conditions are met:
\begin{enumerate}
	\item If $\calF(x,y)=\calF(x_0,y_0)$, then $x\in [x_0]^L_{\scrT}$.
	\item If $\calF(x_0,y)=\calF(x_0,y_0)$, then $(x_0,y)\in [(x_0,y_0)]_{\scrT}$.
\end{enumerate}
\end{theorem}

\begin{proof}
To prove sufficiency, we suppose Conditions 1 and 2 are met. Let $\calF(x,y)=\calF(x_0,y_0)$ for nonzero $x_0,y_0$. Then, by Condition 1, $x\in [x_0]^L_{\scrT}$. Hence, there exists $\calT_1=(\calT_{\calX,1},\calT_{\calY,1})\in\scrT$ such that $x=\calT_{\calX,1}(x_0)$. Therefore $\calF(x_0,y_0) = \calF(x,y)=\calF(\calT_1^{-1}(x,y))=\calF(x_0,\calT_{\calY,1}^{-1}(y))$. By Condition 2, there exists $\calT_2\in\scrT$ such that $(x_0,\calT_{\calY,1}^{-1}(y))=\calT_2(x_0,y_0)$. Hence $(x,y)=\calT_1(x_0,\calT_{\calY,1}^{-1}(y))=\calT_1\circ\calT_2(x_0,y_0)$, and $(x_0,y_0)$ is identifiable up to $\scrT$.

Next we prove necessity. Given that $(x_0,y_0)$ ($x_0\neq 0,y_0\neq 0$) is identifiable up to $\scrT$, by Definition \ref{def:iutg}, if $\calF(x,y)=\calF(x_0,y_0)$, then $(x,y)\in [(x_0,y_0)]_{\scrT}$. The necessity of Conditions 1 and 2 follows.
\end{proof}

\begin{corollary}\label{cor:ibip}
In the constrained BIP, the pair $(x_0,y_0)$ ($x_0\neq 0,y_0\neq 0$) is identifiable up to $\scrT$ if the following two conditions are met:
\begin{enumerate}
	\item If $\calF(x,y)=\calF(x_0,y_0)$, then $x\in [x_0]^L_{\scrT}$.
	\item If $\calF(x_0,y)=\calF(x_0,y_0)$, then $y=y_0$.
\end{enumerate}
Furthermore, if $\calF$ has no non-trivial right annihilator of $\Omega_\calX$, and for $(\calT_\calX,\calT_\calY)\in\scrT$,  $\calT_\calX(x_0)=x_0$ only if $\calT_\calX = \mathbf{1}_\calX$, then the sufficient conditions above are also necessary.
\end{corollary}

\begin{proof}
Given that $y=y_0$, we have that $(x_0,y)=\mathbf{1}(x_0,y_0)$ and hence $(x_0,y)\in[(x_0,y_0)]_{\scrT}$. Therefore, condition 2 in Corollary \ref{cor:ibip} is more demanding than that of Theorem \ref{thm:ibip}. Sufficiency follows. 

The necessity of condition 1 also follows from Theorem \ref{thm:ibip}. Next we show that with the extra assumptions, condition 2 is also necessary. Given that $(x_0,y_0)$ ($x_0\neq 0,y_0\neq 0$) is identifiable up to $\scrT$, by Theorem \ref{thm:ibip}, if $\calF(x_0,y)=\calF(x_0,y_0)$, then there exists $\calT=(\calT_\calX,\calT_\calY)\in\scrT$ such that $(x_0,y)=\calT(x_0,y_0)$. The first argument $\calT_\calX(x_0)=x_0$, by the extra assumption, $\calT_\calX=\mathbf{1}_\calX$. Now, for all $(x_1,y_1)\in \Omega_\calX\times \Omega_\calY$, $\calF(x_1,y_1) = \calF(\calT(x_1,y_1))=\calF(\mathbf{1}_\calX(x_1),\calT_\calY(y_1))=\calF(x_1,\calT_\calY(y_1))$, or equivalently, $\calF(x_1,y_1-\calT_\calY(y_1))=0$. By the extra assumption that $\calF$ has no non-trivial right annihilator of $\Omega_\calX$, $y_1-\calT_\calY(y_1)=0$ for all $y_1\in \Omega_\calY$, or equivalently, $\calT_\calY=\mathbf{1}_\calY$. Therefore, $y=\calT_\calY(y_0)=y_0$, and condition 2 is necessary.
\end{proof}

The extra assumptions in Corollary \ref{cor:ibip} are usually satisfied, which means that Condition 2 is usually also necessary. Indeed, most bilinear mappings that arise in applications have no non-trivial annihilators. The assumption that ``$\calT_\calX(x_0)=x_0$ only if $\calT_\calX=\mathbf{1}_\calX$'' is also true in many scenarios. For example, if $\calT_\calX$ is scaling by a nonzero complex number and $\calT_\calX(x_0)=x_0$ for some nonzero $x_0$, then $\calT_\calX$ has to be identity. 
However, there are examples for which Corollary \ref{cor:ibip} is not necessary (cf. Appendix \ref{app:nta}).


Later in this paper, we repeatedly apply Corollary \ref{cor:ibip} to various scenarios of the blind gain and phase calibration problem and derive sufficient conditions for identifiability up to transformation groups.


\section{BGPC with a Subspace Constraint} \label{sec:subspace}
In this section, we consider the identifiability of the BGPC problem with a subspace constraint. The measurement in the following problem is $Y=\diag(\lambda_0)AX_0$. The known matrix $A\in\bbC^{n\times m}$ is tall ($n>m$). The columns of $\Phi = AX$ reside in a low-dimensional subspace. The constraint sets are $\Omega_\Lambda = \bbC^n$ and $\Omega_\calX=\bbC^{m\times N}$, hence the problem in unconstrained with respect to $\lambda$ and $X$.
\begin{align*}
\text{find}~~&(\lambda,X),\\
\text{s.t.}~~&\diag(\lambda)AX = Y,\\
& \lambda \in \bbC^n,~X \in \bbC^{m\times N}.
\end{align*}

\subsection{Sufficient Condition}
As was mentioned earlier, the BGPC problem suffers from scaling ambiguity. The ambiguity transformation group is defined as follows:
\begin{equation}
\scrT = \{\calT: \calT(\lambda,X)=(\sigma\lambda,\frac{1}{\sigma}X),\text{ for some nonzero $\sigma\in\bbC$}\}.
\label{eq:14}
\end{equation}

Next, we investigate identifiability up to scaling within the framework of Section \ref{sec:bip}. By applying Corollary \ref{cor:ibip}, we provide an alternative proof for the results by Nguyen et al. \cite{Nguyen2013}.
We need the following definition and lemma (See Appendix \ref{app:proofdecomposable} for the proof).
\begin{definition}
The row space of a matrix $A\in\bbC^{n\times m}$ is said to be decomposable if there exists a non-empty proper subset (neither the empty set nor the universal set) $J\subset \{1,2,\cdots,n\}$ and its complement $J^c$ such that $\calR(A)=\calR(A^{(J,:)})\oplus\calR(A^{(J^c,:)})$.
\end{definition}
\begin{lemma}\label{lem:decomposable}
\begin{enumerate}
	\item If $A$ has full row rank, then the row space of $A$ is decomposable.
	\item If $A\in\bbC^{n\times m}$ has full column rank and its row space is not decomposable, then $n>m$.
	\item The row space of $A$ is not decomposable if and only if $\dim(\calR(A))<\dim(\calR(A^{(J,:)}))+\dim(\calR(A^{(J^c,:)}))$ for all non-empty proper subsets $J\subset \{1,2,\cdots,n\}$.
\end{enumerate}
\end{lemma}

Nguyen et al. \cite{Nguyen2013} referred to the property that ``$A$ has full column rank and its row space is not decomposable'' as ``nonseparable full rank''. Here is our restatement of the identifiability result followed by an alternative proof.
\begin{theorem}\label{thm:subspace}
In the BGPC problem with a subspace constraint, the pair $(\lambda_0,X_0)\in \bbC^n\times \bbC^{m\times N}$ is identifiable up to an unknown scaling if the following conditions are met:
\begin{enumerate}
	\item Vector $\lambda_0$ is non-vanishing, i.e., all the entries of $\lambda_0$ are nonzero.
	\item Matrix $X_0$ has full row rank.
	\item Matrix $A$ has full column rank and its row space is not decomposable.
\end{enumerate}
\end{theorem}

\begin{proof}
We apply Corollary \ref{cor:ibip} to the BGPC problem, and verify that the two conditions in the corollary are satisfied. First, since the vector $\lambda_0$ is non-vanishing and the matrix $A$ has full column rank, $\diag(\lambda_0)A$ has full column rank. It follows that if $\diag(\lambda_0)AX_0=\diag(\lambda_0)AX_1$, then $X_1=X_0$. Hence, given $\lambda_0$, the recovery of $X_0$ is unique. This verifies Condition 2 in Corollary \ref{cor:ibip}. To verify Condition 1, we only need to show that $\lambda_0$ is identifiable up to scaling.

We prove by contradiction. Suppose the opposite, that there exists $(\lambda_1,X_1)$ such that $\diag(\lambda_0)A X_0=\diag(\lambda_1)AX_1$ but $\lambda_1 \notin [\lambda_0]^L_{\scrT}$. Recall that all the entries of $\lambda_0$ are nonzero, $A$ has full column rank and $X_0$ has full row rank. Therefore, $\rank(\diag(\lambda_0)A X_0)=\rank(\diag(\lambda_1)AX_1)=m$, and $X_1$ too has full row rank. Since the row space of $A$ is not decomposable, there are no zero rows in $A$. Because $X_0$ and $X_1$ have full row rank, it follows that there are no zero rows in $AX_0$ or $AX_1$. The vector $\lambda_0$ is non-vanishing, hence $\lambda_1$ too is non-vanishing. Let $\gamma=\lambda_1./\lambda_0$ denote the entrywise ratio of $\lambda_1$ over $\lambda_0$, where $\gamma^{(j)}={\lambda_1^{(j)}}/{\lambda_0^{(j)}}\neq 0, j=1,2,\cdots,n$. By the assumption that $\lambda_1\notin [\lambda_0]^L_{\scrT}$, the entrywise ratio is not the repetition of the same number, i.e., there exist $j_1,j_2$ such that $\gamma^{(j_1)}\neq \gamma^{(j_2)}$. Let $T$ denote the number of distinct values of $\gamma^{(j)}$. Create a partition of the index set $\{1,2,\cdots,n\}$, denoted by $J_1,J_2,\cdots,J_T$, such that $\gamma^{(j)}=\gamma_t$ for all $j\in J_{t}$, $t=1,2,\cdots,T$. Note that $\gamma_1,\gamma_2,\cdots,\gamma_T$ are the distinct values of $\gamma^{(j)}$.

Consider the row spaces of $A$:
\begin{equation}
\calR(A) = \sum_{t=1}^{T}{\calR(A^{(J_t,:)})}.
\label{eq:2}
\end{equation}
Denote the dimension of $\calR(A^{(J_t,:)})$ by $m_t$. Then there exists a subset $J_t^b\subset J_t$ such that $|J_t^b|=m_t$ and the rows of $A^{(J_t^b,:)}$ form a basis of $\calR(A^{(J_t,:)})$. By the condition that the row space of $A$ is not decomposable, the sum in \eqref{eq:2} is not a direct sum, hence $m=\rank(A) < \sum_{t=1}^{T}m_t$. Furthermore, by \eqref{eq:2}, there exists a subset
\[
J^b=\{j_1,j_2,\cdots,j_m\}\subset \bigcup_{t=1}^{T}J_t^b \subset \bigcup_{t=1}^{T}J_t = \{1,2,\cdots,n\}
\]
such that $|J^b|=m$ and the rows of $A^{(J^b,:)}$ form a basis of $\calR(A)$.
The set $(\bigcup_{t=1}^{T}J_t^b)\setminus J^b$ is not empty because $m<\sum_{t=1}^{T}{m_t}$. Without loss of generality, we may assume that there exists $j_0\in J_1^b\setminus J^b$. The row $A^{(j_0,:)}$ can be written as a linear combination of the rows of $A^{(J^b,:)}$ and the representation is unique. We denote the representation by:
\begin{equation}
A^{(j_0,:)}=\alpha_{j_1}A^{(j_1,:)}+\alpha_{j_2}A^{(j_2,:)}+\cdots+\alpha_{j_m}A^{(j_m,:)}. 
\label{eq:3}
\end{equation}
The rows of $A^{(J^b_1,:)}$ are linearly independent, and $j_0\notin J^b_1\bigcap J^b$, hence $A^{(j_0,:)}$ cannot be written as a linear combination of the rows of $A^{(J^b_1\bigcap J^b,:)}$; there exists at least one nonzero term in the representation \eqref{eq:3} corresponding to one of the rows of $A^{(J^b\setminus J^b_1,:)}$. Thus, without loss of generality, there exists $j_1\in J^b\bigcap J_2^b$, such that $\alpha_{j_1}\neq 0$. 

Recall that $\rank(\diag(\lambda_0)A X_0)=\rank(\diag(\lambda_1)AX_1)=m$, and $X_0$ and $X_1$ have full row rank $m$. Therefore, the column spaces satisfy:
\[
\calC(\diag(\lambda_0)A) = \calC(\diag(\lambda_0)A X_0) = \calC(\diag(\lambda_1)AX_1)= \calC(\diag(\lambda_1)A).\]
Hence $\rank([\diag(\lambda_0)A,\diag(\lambda_1)A]) = m$. Defining matrix
\[
B \coloneqq [A,\diag(\gamma)A] = [\diag(\lambda_0)]^{-1} [\diag(\lambda_0)A,\diag(\lambda_1)A].
\]
We have that
\begin{equation}
\rank(B)=\rank([\diag(\lambda_0)A,\diag(\lambda_1)A]) = m.
\label{eq:1}
\end{equation}
Then we consider the row spaces of $B$. The dimension of the row space $\calR(B^{(J_t,:)})=\calR([A^{(J_t,:)},\gamma_t A^{(J_t,:)}])$ is also $m_t$, and the rows of $B^{(J^b_t,:)}$ form a basis of the above row space. The rows of $B^{(J^b,:)}$ form a linearly independent set of cardinality $m$. By $\eqref{eq:1}$, the rows of $B^{(J^b,:)}$ form a basis of $\calR(B)$.
The row $B^{(j_0,:)}$ can be can be written as a linear combination of the rows of $B^{(J^b,:)}$. We denote the representation by:
\begin{equation}
B^{(j_0,:)}=\beta_{j_1}B^{(j_1,:)}+\beta_{j_2}B^{(j_2,:)}+\cdots+\beta_{j_m}B^{(j_m,:)}.
\label{eq:4}
\end{equation}
Recall that $j_0\in J^b_1$, $j_1\in J^b_2$, hence $\gamma^{(j_0)} = \gamma_1$, and $\gamma^{(j_1)}= \gamma_2$. Using the definition of $B$, we rewrite \eqref{eq:4} as:
\begin{align}
A^{(j_0,:)}&=\beta_{j_1}A^{(j_1,:)}+\beta_{j_2}A^{(j_2,:)}+\cdots+\beta_{j_m}A^{(j_m,:)},
\label{eq:5} \\
\gamma_1 A^{(j_0,:)}&=\beta_{j_1}\gamma_2A^{(j_1,:)}+\beta_{j_2}\gamma^{(j_2)}A^{(j_2,:)}+\cdots+\beta_{j_m}\gamma^{(j_m)}A^{(j_m,:)}.
\label{eq:6}
\end{align}
Since the representation in \eqref{eq:3} is unqiue, the representations in \eqref{eq:5} and \eqref{eq:6} must satisfy:
\[
\beta_{j_1} = \beta_{j_1}\frac{\gamma_2}{\gamma_1}=\alpha_{j_1}\neq 0.
\]
It follows that $\gamma_1=\gamma_2$, which contradicts the assumption that $\gamma_1$ and $\gamma_2$ are distinct. Hence the assumption that $\lambda_1\notin [\lambda_0]^L_{\scrT}$ is false, and $\lambda_0$ is identifiable up to an unknown scaling.
\end{proof}

For generic signals, we can show that Theorem \ref{thm:subspace} reduces to a simple condition (Corollary \ref{cor:subspace}) on the dimensions $n$, $m$ and $N$. We say that a property holds for almost all signals if the property holds for all signals but a set of measure zero.
\begin{corollary}\label{cor:subspace}
In the BGPC problem with a subspace constraint, if $n>m$ and $N\geq m$, then $(\lambda_0,X_0)$ is identifiable up to an unknown scaling for almost all $\lambda_0\in\bbC^n$, almost all $X_0\in\bbC^{m\times N}$ and almost all $A\in\bbC^{n\times m}$.
\end{corollary}
\begin{proof}
Almost all $\lambda_0\in\bbC^n$ are non-vanishing.
If $N\geq m$, almost all $X_0\in\bbC^{m\times N}$ have full row rank.
If $n>m$, almost all $A\in\bbC^{n\times m}$ have full column rank. Next we show that the row spaces of almost all $A$ are not decomposable. For almost all $A$, the submatrices $A^{(J,:)}$ and $A^{(J^c,:)}$ have full rank for every non-empty proper subset $J\subset \{1,2,\cdots,n\}$. Therefore, one of the following cases has to be true.
\begin{enumerate}
	\item If $|J|<m$ and $|J^c|<m$, then for almost all $A$, $\dim(\calR(A)) =  m$, $\dim(\calR(A^{(J,:)})) = |J|$, $\dim(\calR(A^{(J^c,:)}))=|J^c|$. Hence for almost all $A$,
	\[
	\dim(\calR(A)) =  m < n = |J|+|J^c| = \dim(\calR(A^{(J,:)}))+\dim(\calR(A^{(J^c,:)})).
	\]
	\item If $|J|\geq m$, then for almost all $A$, $\dim(\calR(A^{(J,:)}))=m$. Hence for almost all $A$,
	\[
	\dim(\calR(A)) =  m < m+1 \leq \dim(\calR(A^{(J,:)}))+\dim(\calR(A^{(J^c,:)})).
	\]
	\item If $|J^c|\geq m$, then for almost all $A$, $\dim(\calR(A^{(J^c,:)}))=m$. Hence for almost all $A$,
	\[
	\dim(\calR(A)) =  m < 1+m \leq \dim(\calR(A^{(J,:)}))+\dim(\calR(A^{(J^c,:)})).
	\]
\end{enumerate}
Therefore, $\dim(\calR(A))<\dim(\calR(A^{(J,:)}))+\dim(\calR(A^{(J^c,:)}))$ for every non-empty proper subset $J\subset\{1,2,\cdots,n\}$, establishing that the row spaces of almost all $A$ are not decomposable.
By Theorem \ref{thm:subspace}, given that $N\geq m$ and $n>m$, the pair $(\lambda_0,X_0)$ is identifiable up to an unknown scaling for almost all $\lambda_0$, $X_0$ and $A$.
\end{proof}
Corollary \ref{cor:subspace} shows that, in the BGPC problem with a subspace constraint, for almost all vectors $\lambda_0$, almost all tall matrices $A$ and almost all fat matrices $X_0$, the solution $(\lambda_0,X_0)$ is identifiable up to an unknown scaling.

\subsection{Necessary Condition}
Given that $\lambda_0$ is non-vanishing, Nguyen et al. \cite{Nguyen2013} showed that ``the row space of $A$ is not decomposable'' is necessary.
Lacking however, is a necessary condition for the sample complexity.

As we demonstrate in the next subsection by construction of counter-examples, the sample complexity $N\geq m$, as required by Theorem \ref{thm:subspace} implicitly and Corollary \ref{cor:subspace} explicitly, is not necessary. Instead, a necessary condition is suggested by heuristically counting the number of degrees of freedom and the number of measurements in $Y=\diag(\lambda)AX$. The numbers of free variables in $\lambda$ and $X$ are $n$ and $mN$, respectively. The unknown scaling of $\lambda$ and $X$ is counted twice, hence $1$ is subtracted yielding $n+mN-1$ for the total number of degrees of freedom. The total number of measurements is $nN$. Heuristically, to achieve uniqueness, $nN$ must be greater than or equal to $n+mN-1$, which implies $N\geq\frac{n-1}{n-m}$. This turns out to be a valid necessary condition, as we now state and prove rigorously.

\begin{proposition}\label{pro:necessary}
In the BGPC problem with a subspace constraint, if $A$ has full column rank, and $(\lambda_0,X_0)$ (with a non-vanishing $\lambda_0$) is identifiable up to scaling, then $N\geq\frac{n-1}{n-m}$.
\end{proposition}

\begin{proof}
We show that if $N<\frac{n-1}{n-m}$, then the recovery cannot be unique.
Let $A_\perp\in\bbC^{n\times (n-m)}$ denote a matrix whose columns form a basis for the ortho-complement of the column space of $A$. Hence $A_\perp^*$ is an annihilator of the column space of $A$. Consider the linear operator $\mathcal{G}:\bbC^n\rightarrow\bbC^{(n-m)\times N}$ defined by
\begin{equation*}
\mathcal{G}(x) \coloneqq A_\perp^*\diag(x)Y = A_\perp^*\diag(x)\diag(\lambda_0)AX_0.
\end{equation*}
We claim that every non-vanishing null vector of $\mathcal{G}$ produces a solution to the BGPC problem. Indeed, if $x\in\mathcal{N}(\mathcal{G})$, then \[A_\perp^*\diag(x)\diag(\lambda_0)AX_0 = 0,\] 
hence the columns in $\diag(x)\diag(\lambda_0)AX_0$ must reside in the column space of $A$. Let \[\diag(x)\diag(\lambda_0)AX_0=AX_1.\] If $x$ is non-vanishing, then $(\lambda_1,X_1)$ is a solution, where $\lambda_1$ is the entrywise inverse of $x$.

Let $x_0$ denote the entrywise inverse of $\lambda_0$, then $x_0\in\mathcal{N}(\mathcal{G})$. There are $N(n-m)$ equations in $\mathcal{G}(x)=0$. If $N<\frac{n-1}{n-m}$, i.e., $N(n-m)\leq n-2$, the dimension of the null space $\mathcal{N}(\mathcal{G})$ is at least $2$. Hence, there exists another vector $x_1\in\mathcal{N}(\mathcal{G})$ such that $x_0,x_1$ are linearly independent. Let $\alpha$ be a complex number such that $0<|\alpha|<\frac{1}{\norm{\lambda_0}_\infty\norm{x_1}_\infty}$. Then $x_0+\alpha x_1\in\mathcal{N}(\mathcal{G})$ is non-vanishing, because the entries of $x_0+\alpha x_1$ satisfy that
\[
\bigl|x_0^{(j)}+\alpha x_1^{(j)}\bigr|\geq\bigl|x_0^{(j)}\bigr|-|\alpha|\bigl|x_1^{(j)}\bigr|\geq \frac{1}{\norm{\lambda_0}_\infty}-|\alpha|\norm{x_1}_\infty>0,\quad \text{for every $j\in\{1,2,\cdots,n\}$.}
\]
This null vector is not a scaled version of $x_0$. Hence there exists a solution that does not belong to the equivalence class $[(\lambda_0,X_0)]_\scrT$. Therefore, $N\geq \frac{n-1}{n-m}$ is necessary.
\end{proof}

The two sample complexities $N\geq m$ and $N\geq \frac{n-1}{n-m}$ coincide when $m = 1$ or $m = n-1$. The gap between these two sample complexities when $2\leq m \leq n-2$ is analyzed next.

\subsection{Gap Between the Sufficient and the Necessary Conditions}\label{sec:gap1}
The sample complexity in the sufficient condition is $N\geq m$, which can be represented by the region to the right of a line segment. The sample complexity in the necessary condition is $N\geq \frac{n-1}{n-m}$, which can be represented by the region to the right of part of a hyperbola. The gap between the two sample complexities is the region between the line segment and the hyperbola (cf. Figure \ref{fig:ptss}).
\begin{figure}[htbp]%
\centering
\input{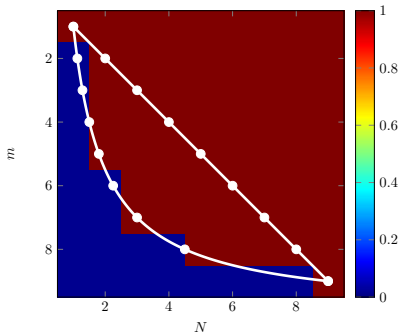}
\caption{The sample complexities for BGPC with a subspace constraint, and the ratio of identifiable pairs generated randomly.}%
\label{fig:ptss}%
\end{figure}

To explore this gap, we wish to determine whether $(\lambda_0,X_0)$, in BGPC with a subspace constraint, is identifiable up to scaling. We now show that this can be done by Algorithm 1. Given $A$ that has full column rank and $Y=\diag(\lambda_0)AX_0$ that has no zero rows, Algorithm 1 returns a boolean value indicating whether $(\lambda_0,X_0)$ is identifiable up to scaling. 
\begin{algorithm}[H]
\caption{Identifiability of the BGPC problem with a subspace constraint}
\begin{algorithmic}
\STATE \textbf{input}: $A$,$Y$ \quad \textbf{output}: identifiability of $(\lambda_0,X_0)$
\STATE $[Q,R] = \texttt{qr}(A)$ \COMMENT{\texttt{QR decomposition of $A$}}
\STATE $A_\perp \leftarrow Q^{(:,m+1:n)}$
\STATE $G \leftarrow \begin{bmatrix}[\diag(Y^{(:,1)})]^*A_\perp & [\diag(Y^{(:,2)})]^*A_\perp & \cdots & [\diag(Y^{(:,N)})]^*A_\perp \end{bmatrix}^*$
\IF{$\rank(G)\leq n-2$}
\RETURN False
\ELSE
\RETURN True
\ENDIF
\end{algorithmic}
\end{algorithm}

\begin{proposition}
Given $A$ that has full column rank and $Y=\diag(\lambda_0)AX_0$ that has no zero rows, the pair $(\lambda_0,X_0)$ is identifiable up to scaling if Algorithm 1 returns True, and not identifiable up to scaling if Algorithm 1 returns False.
\end{proposition}
\begin{proof}
The columns of $A_\perp$ form a basis for the ortho-complement of the column space of $A$, hence $A_\perp^*$ is an annihilator of the column space of $A$. The matrix $G\in\bbC^{N(n-m)\times n}$ satisfies that $Gx=\operatorname{vec}(A_\perp^*\diag(x)Y)$.
Given $Y$ that has no zero rows, any solution to the BGPC problem $(\lambda,X)$ satisfies that $\lambda$ is non-vanishing, and that the entrywise inverse of $\lambda$ is a null vector of $G$. On the other hand, as argued in the proof of Proposition \ref{pro:necessary}, any non-vanishing null vector of $G$ produces a solution $(\lambda,X)$.

If Algorithm 1 returns True, then $\rank(G)\geq n-1$. Given a solution $(\lambda_0,X_0)$, $G$ has at least one null vector $x_0$, which is the entrywise inverse of $\lambda_0$. Hence $\rank(G)= n-1$. All the null vectors of $G$ reside in the one-dimensional subspace spanned by $x_0$. Therefore $\lambda$ in any solution is a scaled version of $\lambda_0$, or $\lambda\in[\lambda_0]^L_\scrT$. Given non-vanishing $\lambda_0$ and $A$ with full column rank, $\diag(\lambda_0)A$ has full column rank and the recovery of $X_0$ has to be unique. By Corollary \ref{cor:ibip}, $(\lambda_0,X_0)$ is identifiable up to scaling.

If Algorithm 1 returns False, then $\rank(G)\leq n-2$. By the proof of Proposition \ref{pro:necessary}, $(\lambda_0,X_0)$ is not identifiable.
\end{proof}

We now use Algorithm 1 to construct counter-examples demonstrating that the sufficient condition in Theorem \ref{thm:subspace} is not necessary. Let $n=10$, $1\leq m\leq 9$, and $1\leq N\leq 9$. The entries of $\lambda_0\in\bbR^{n}$ and $X_0\in\bbR^{m\times N}$ are generated as iid Gaussian random variables $N(0,1)$. The matrix $A\in\bbR^{n\times m}$ is the first $m$ columns from an $n\times n$ random orthogonal matrix. Then $A_\perp$ comprises of the last $(n-m)$ columns from the same random orthogonal matrix. We use Algorithm 1 to determine whether or not $(\lambda_0,X_0)$ is identifiable up to scaling. For every value of $m$ and $N$, the numerical experiment is repeated 100 times independently. The ratio of identifiable pairs as a function of $(m,N)$ is shown in Figure \ref{fig:ptss}. As is expected, the solution $(\lambda_0,X_0)$ is identifiable when $N\geq m$, and is not identifiable when $N< \frac{n-1}{n-m}$. Meanwhile, when $\frac{n-1}{n-m}\leq N< m$, the ratio of identifiable pairs is $1$. Therefore, $N\geq m$ is not necessary. 

On the other hand, the necessary condition in Proposition \ref{pro:necessary} is not sufficient. For example, if $n=8$, $m=4$ and $\frac{n-1}{n-m}<N=2<m$, let $A$ be the structured matrix
\[
A =
\begin{bmatrix}
A_1 & \diag(\gamma)A_1
\end{bmatrix},
\]

where $A_1\in\bbC^{8\times 2}$, $\gamma\in\bbC^{8}$. There exists an $A_1$ and a $\gamma$ such that the matrix $A$ has full column rank and the row space of $A$ is not decomposable. For example, let $A_1=2\sqrt{2}F^{(:,1:2)}$ and $\gamma=2\sqrt{2}F^{(:,3)}$, then $A=2\sqrt{2}F^{(:,1:4)}$. However, $(\lambda_0,X_0)$ is not identifiable and $\lambda_0 AX_0=\lambda_1 AX_1$, if
\[
X_0 = \begin{bmatrix} 1 & 0 \\ 0 & 1 \\ 0 & 0 \\ 0 & 0 \end{bmatrix}, \quad
X_1 = \begin{bmatrix} 0 & 0 \\ 0 & 0 \\ 1 & 0 \\ 0 & 1 \end{bmatrix},
\]
\[
\lambda_0 = \gamma\odot \lambda_1. 
\]

However, according to the ratio of identifiable pairs shown in Figure \ref{fig:ptss}, the unidentifiable case does not occur even once in 100 random trials. We have the following conjecture:
\begin{conjecture}
In the BGPC problem with a subspace constraint, if $n>m$ and $N\geq \frac{n-1}{n-m}$, then $(\lambda_0,X_0)$ is identifiable up to an unknown scaling for almost all $\lambda_0\in\bbC^n$, almost all $X_0\in\bbC^{m\times N}$ and almost all $A\in\bbC^{n\times m}$.
\end{conjecture}
If the above conjecture is true, the necessary condition $N\geq \frac{n-1}{n-m}$ is tight except for a set of measure zero.

\section{BGPC with a Joint Sparsity Constraint}\label{sec:jointsparsity}
Here we consider the identifiability in the BGPC problem with a joint sparsity constraint:
\begin{align*}
\text{($\mathrm{P1}$)}\quad\text{find}~~&(\lambda,X),\\
\text{s.t.}~~&\diag(\lambda)AX = Y,\\
& \lambda \in \bbC^n,~X \in \Omega_\calX =\{X\in\bbC^{n\times N}: \text{the columns of $X$ are jointly $s$-sparse}\}.
\end{align*}
The measurement in the above problem is $Y=\diag(\lambda_0)AX_0$. We only consider the case where $A\in\bbC^{n\times n}$ is an invertible square matrix. The vector $\lambda_0\in\bbC^n$ is non-vanishing. The columns of $X_0\in\bbC^{n\times N}$ are jointly $s$-sparse ($X_0$ has at most $s$ nonzero rows).
Unless otherwise stated, we assume that the sparsity level $s$ is known a priori. However, if $s$ is unknown, one can solve the following optimization problem instead:
\begin{align*}
\text{($\mathrm{P2}$)}\quad\underset{(\lambda,X)}{\text{min.}}~~&\operatorname{row-sparsity}(X),\\
\text{s.t.}~~&\diag(\lambda)AX = Y,\\
& \lambda \in \bbC^n,~X \in\bbC^{n\times N}.
\end{align*}

In this section, we define ambiguities and transformation groups that depend on the matrix $A$. For two special cases of $A$, we give sufficient conditions for identifiability up to the ambiguity transformation groups.

\subsection{Ambiguities and Transformation Groups}\label{sec:equivalence}
Geometrically, a joint sparsity constraint corresponds to a union of subspaces; hence, it is less restrictive than the previously discussed subspace constraint. This results in greater ambiguity in identifying a solution to BGPC with a joint sparsity constraint, than just the scaling ambiguity. In this case, to obtain identifiablity results, we must choose the largest transformation group associated with the BIP, which captures \emph{all} ambiguities inherent in the problem. In this section, we develop a procedure to do so.

A generalized permutation matrix is an invertible square matrix with exactly one nonzero entry in each row and each column. It preserves the joint sparsity structure. That is, if the columns of $X_0$ are jointly $s$-sparse and $P$ is a generalized permutation matrix, then the columns of $X_1=PX_0$ are also jointly $s$-sparse. Suppose there exists a vector $\gamma\in \bbC^n$ such that $P = A^{-1}\diag(\gamma)A$ is a generalized permutation matrix; then clearly $\gamma$ has to be non-vanishing. Now, given a solution $(\lambda_0,X_0)$ to the BGPC problem, there exist $\lambda_1=\lambda_0./\gamma$ and $X_1 = PX_0\in \Omega_\calX$ such that
\[
\diag(\lambda_1)AX_1 = \diag(\lambda_0)[\diag(\gamma)]^{-1}AA^{-1}\diag(\gamma)AX_0=\diag(\lambda_0)AX_0.
\]
This ambiguity is inevitable. To address this ambiguity, we define the set
\begin{equation}
\Gamma(A) = \{\gamma\in\bbC^n:A^{-1}\diag(\gamma)A \text{ is a generalized permutation matrix}\},
\label{eq:7}
\end{equation}
and the ambiguity transformation group
\begin{equation}
\scrT = \{\calT:\calT(\lambda,X)=(\lambda./\gamma,A^{-1}\diag(\gamma)AX) \text{ for some $\gamma\in \Gamma(A)$}\}.
\label{eq:8}
\end{equation}
Then $(\lambda_1,X_1)$ is in the equivalence class $[(\lambda_0,X_0)]_{\scrT}$. 

Note that the set $\Gamma(A)$ depends on $A$. In particular, when $A$ is the normalized DFT matrix $A=F\in \bbC^{n\times n}$, the matrix $F^*\diag(\gamma)F$ is a circulant matrix whose first column is $\frac{1}{\sqrt{n}}F^*\gamma$. The matrix $F^*\diag(\gamma)F$ is a generalized permutation matrix if and only if there is exactly one nonzero entry in $\frac{1}{\sqrt{n}}F^*\gamma$, which means that the circulant matrix $F^{*}\diag(\gamma)F$ is a scaled circular shift. Therefore,
\begin{equation}
\Gamma(F) = \bigl\{\gamma=\sigma \sqrt{n}F^{(:,k)}: \sigma\in\bbC \text{ is nonzero}, k\in\{1,2,\cdots,n\}\bigr\}.
\label{eq:9}
\end{equation}
\begin{equation}
\scrT = \{\calT:\calT(\lambda,X)=(\lambda./\gamma,F^{*}\diag(\gamma)FX) \text{ for some $\gamma\in \Gamma(F)$}\}.
\label{eq:10}
\end{equation}
An equivalence transformation $\calT\in\scrT$ defined in \eqref{eq:10} is a complex exponential modulation of $\lambda$ scaled by $\frac{1}{\sigma}$ and a circular shift of $X$ scaled by $\sigma$. In MBD, if we shift the signal by $1-k$ and scale it by $\frac{1}{\sigma}$, and shift the channels by $k-1$ and scale them by $\sigma$, the outputs of the channels remain unchanged.

The ambiguity transformation groups for other choices of $A$ can be figured out in a similar fashion. For more examples, please refer to Section \ref{sec:piecewiseconstant} and to Appendix \ref{app:tgforA}.


\subsection{Identifiability of Jointly Sparse Signals}\label{sec:js}
In this section, we assume that $A=F$ is the DFT matrix and the columns of $X$ are jointly $s$-sparse. In multichannel blind deconvolution, the non-vanishing vector $\lambda_0$ is the DFT of the signal and the jointly sparse columns of $X_0$ are the multiple channels. We derive a sufficient condition and a necessary condition for $(\lambda_0,X_0)$ to be identifiable up to the transformation group defined in \eqref{eq:10}.

\subsubsection{Sufficient Condition}
We can prove a sufficient condition for identifiability up to the transformation group in \eqref{eq:10} within the framework of Section \ref{sec:bip} by again invoking Corollary \ref{cor:ibip}.
We need the following definition to state this sufficient condition.
\begin{definition}
The index set $J=\{j_1,j_2,\cdots,j_s\}\subset\{1,2,\cdots,n\}$ is said to be periodic with period $\ell$ ($\ell$ being an integer such that $0<\ell<n$), if $J=\{j_1+\ell,j_2+\ell,\cdots,j_s+\ell\}$ (modulo $n$). The smallest integer $\ell$ with this property is called the fundamental period.
\end{definition}

The universal set $\{1,2,\cdots,n\}$ is always periodic with period $\ell$ ($\ell$ being any integer from $1$ to $n-1$). The fundamental period is $1$. For $n=10$ and $s=4$, the set $J=\{1,2,6,7\}$ is periodic with fundamental period $5$. Periodicity has the following property.

\begin{remark}\label{rem:periodic}
If the set $J=\{j_1,j_2,\cdots,j_s\}$ is periodic with period $\ell$, then the complement $J^c$, the flipped version $J^f=\{-j_1,-j_2,\cdots,-j_{s}\}$ (modulo $n$) and the shifted version $\{j_1+k,j_2+k,\cdots,j_s+k\}$ (modulo $n$) are all periodic with period $\ell$.
\end{remark}

Here is the sufficient condition for the identifiability of the BGPC problem with DFT matrix and a joint sparsity constraint.
\begin{theorem}\label{thm:jointsparsity}
In the BGPC problem with DFT matrix and a joint sparsity constraint at sparsity level $s$, the pair $(\lambda_0,X_0)\in \bbC^n\times \Omega_\calX$ is identifiable up to the transformation group $\scrT$ defined in \eqref{eq:10} if the following conditions are met:
\begin{enumerate}
	\item Vector $\lambda_0$ is non-vanishing.
	\item Matrix $X_0$ has exactly $s$ nonzero rows and rank $s$.
	\item The joint support of the columns of $X_0$ is not periodic.
\end{enumerate}
\end{theorem}

\begin{proof}
First, given non-vanishing $\lambda_0$ and the DFT matrix $F$, the matrix $\diag(\lambda_0)F$ has full rank. If $\diag(\lambda_0)FX_0=\diag(\lambda_0)FX_1$, then $X_1=X_0$. Hence, given $\lambda_0$, the recovery of $X_0$ is unique. By Corollary \ref{cor:ibip}, to complete the proof, we only need to show that $\lambda_0$ is identifiable up to the transformation group.

By assumption, the matrix $X_0$ has rank $s$ and the joint support of the columns of $X_0$, denoted by $J=\{j_1,j_2,\cdots,j_s\}$, is not periodic. Given that $\diag(\lambda_0)FX_0=\diag(\lambda_1)FX_1$, we show that $\lambda_1\in[\lambda_0]^L_\scrT$. Now, the matrix $X_0$ has $s$ linearly independent columns, $\diag(\lambda_0)F$ has full rank, hence the corresponding columns of $X_1$ are also linearly independent. Without loss of generality, we may assume that $X_0$ and $X_1$ only have $s$ columns, which are linearly independent, by removing redundant columns at the same locations in both matrices. Then $X_0,X_1\in\bbC^{n\times s}$ have full column rank $s$ and exactly $s$ nonzero rows. Because $F$ has no zero entries, it follows that there are no zero rows in $FX_0$ or $FX_1$. The vector $\lambda_0$ is non-vanishing, hence $\lambda_1$ is also non-vanishing. We know that
\begin{equation}
P=F^*[\diag(\lambda_1)]^{-1}\diag(\lambda_0)F
\label{eq:11}
\end{equation}
is a circulant matrix and that $X_1=PX_0$.
Let $X_0^{\dagger}\in\bbC^{s\times n}$ denote the pseudo-inverse (also the left inverse) of $X_0$, and $X_{0\perp}\in\bbC^{n\times(n-s)}$ denote a matrix whose columns form a basis for the ortho-complement of the column space of $X_0$. Since $X_0$ has full column rank $s$ and exactly $s$ nonzero rows indexed by $J$, we may choose $X_0^{\dagger}$ such that its nonzero columns are indexed by $J$, and choose the columns of $X_{0\perp}$ to be the standard basis vectors $\{I^{(:,k)}: k\in J^c\}$. The matrix $P$ as in $X_1=PX_0$ satisfies
\begin{equation}
P = X_1X_0^{\dagger}+QX_{0\perp}^*,
\label{eq:12}
\end{equation}
where $Q\in \bbC^{n\times (n-s)}$ is a free matrix. Note that the nonzero columns of $QX_{0\perp}^*$ are indexed by $J^c$ and the nonzero columns of $X_1X_0^{\dagger}$ are indexed by $J$. Hence $P^{(:,J)}=X_1X_0^{\dagger(:,J)}$. The submatrix $P^{(:,J)}$ has no more than $s$ nonzero rows because $X_1$ has $s$ nonzero rows.

We prove $\lambda_1\in[\lambda_0]^L_\scrT$ by contradiction. Suppose that $\lambda_1\notin [\lambda_0]^L_\scrT$. By \eqref{eq:9} and \eqref{eq:10}, the entrywise ratio $\gamma = \lambda_0./\lambda_1\notin \Gamma(F)$, which means that $\frac{1}{\sqrt{n}}F^*\gamma$, the first column of the circulant matrix $P$ (as in \eqref{eq:11}), has more than one nonzero entry. Denote the indices of the first two nonzero entries of $P^{(:,1)}$ by $k_1$ and $k_2$. By the structure of circulant matrices, the rows of $P^{(:,J)}$ indexed by the following two sets (interpreted modulo $n$) are nonzero:
\begin{align*}
K_1 &=\{k_1+j_1-1,k_1+j_2-1,\cdots,k_1+j_s-1\},\\
K_2 &=\{k_2+j_1-1,k_2+j_2-1,\cdots,k_2+j_s-1\}.
\end{align*}
Note that $|K_1|=|K_2|=s$. Recall that $P^{(:,J)}$ has no more than $s$ nonzero rows, hence $K_1=K_2$. It follows that set $K_1$ is periodic with period $\ell = |k_2-k_1|$. By the property in Remark \ref{rem:periodic}, the set $J$ is also periodic with the same period, and we reach a contradiction. Therefore, the assumption that $\lambda_1\notin [\lambda_0]^L_\scrT$ is false, and Condition 1 of Corollary \ref{cor:ibip} is satisfied - the vector $\lambda_0$ is identifiable up to the transformation group.
\end{proof}

\begin{corollary}\label{cor:jointsparsity}
If $N\geq s$, then the conclusion of Theorem \ref{thm:jointsparsity} holds for almost all $\lambda_0\in\bbC^n$, and almost all $X_0\in\bbC^{n\times N}$ that has $s$ nonzero rows and non-periodic joint support.
\end{corollary}
\begin{proof}
Almost all $\lambda_0\in\bbC^n$ are non-vanishing.
If $N\geq s$, then almost all $X_0\in\bbC^{n\times N}$ with $s$ nonzero rows have rank $s$. In addition, the joint support of $X_0$ is not periodic. Therefore, the conditions in Theorem \ref{thm:jointsparsity} are met, and $(\lambda_0,X_0)$ is identifiable up to the transformation group $\scrT$ defined in \eqref{eq:10}.
\end{proof}

Corollary \ref{cor:jointsparsity} shows that, in the BGPC problem with DFT matrix and a joint sparsity constraint, given that $N\geq s$, the identifiability of generic signals $(\lambda_0,X_0)$ hinges on the joint support of $X_0$. If the joint support is non-periodic, $(\lambda_0,X_0)$ is almost always identifiable. Other priors may imply non-periodicity. For example, if the joint support is a contiguous block, or if $n$ and $s$ are coprime, the joint support has to be non-periodic.

\begin{corollary}\label{cor:js1}
If $N\geq s$, then the conclusion of Theorem \ref{thm:jointsparsity} holds for almost all $\lambda_0\in\bbC^n$, and almost all $X_0\in\bbC^{n\times N}$ that has $s$ nonzero rows that are contiguous.
\end{corollary}

\begin{corollary}\label{cor:js2}
If $N\geq s$, and $n$ and $s$ are coprime, then the conclusion of Theorem \ref{thm:jointsparsity} holds for almost all $\lambda_0\in\bbC^n$, and almost all $X_0\in\bbC^{n\times N}$ that has $s$ nonzero rows.
\end{corollary}

Clearly, the coprimeness condition in Corollary \ref{cor:js2} is satisfied for all $s<n$ if $n$ is a prime number.

The above results are under the assumption that the sparsity level $s$ is known a priori. If $s$ is unknown, instead of solving the feasibility problem ($\mathrm{P1}$), one can solve the optimization problem ($\mathrm{P2}$). We have the following corollary, whose proof is almost identical to that of Theorem \ref{thm:jointsparsity}.
\begin{corollary}\label{cor:unknowns1}
In the BGPC problem with DFT matrix and unknown sparsity level, the pair $(\lambda_0,X_0)\in \bbC^n\times \Omega_\calX$ is the unique minimizer of ($\mathrm{P2}$) up to the transformation group $\scrT$ defined in \eqref{eq:10}, if the following conditions are met:
\begin{enumerate}
	\item Vector $\lambda_0$ is non-vanishing.
	\item Matrix $X_0$ has rank equal to the number of nonzero rows.
	\item The joint support of the columns of $X_0$ is not periodic.
\end{enumerate}
\end{corollary}
We can derive row sparsity minimization analogs of Corollaries \ref{cor:jointsparsity}, \ref{cor:js1} and \ref{cor:js2} in a similar fashion. These results are omitted for the sake of brevity.

\subsubsection{Necessary Condition}
Given that $\lambda_0$ is non-vanishing, ``the joint support of the columns of $X_0$ is not periodic'' is necessary. We prove this by contraposition. We assume that the joint support of the columns of $X_0$ is periodic with period $\ell$, and next show that $(\lambda_0,X_0)$ is not identifiable up to the transformation group in \eqref{eq:10}. Let $P$ be a circulant matrix whose first column has two nonzero entries $P^{(1,1)}=1$ and $P^{(\ell+1,1)}=2$. Thus, the DFT $\gamma = \sqrt{n}FP^{(:,1)}$ of the first column of $P$ is non-vanishing. Let $\lambda_1=\lambda_0./\gamma$ and $X_1=PX_0$. Then $P$ satisfies \eqref{eq:11}, and $\diag(\lambda_1)FX_1=\diag(\lambda_0)FX_0$. Since $P$ is not a generalized permutation matrix, $X_1$ is not a scaled and circularly shifted version of $X_0$. Hence $(\lambda_0,X_0)$ is not identifiable up to the transformation group in \eqref{eq:10}.

The above necessary condition does not address the sample complexity. Like Proposition \ref{pro:necessary}, we have the following necessary condition for the sample complexity.
\begin{proposition}\label{pro:necessary2}
In the BGPC problem with DFT matrix and a joint sparsity constraint, if $(\lambda_0,X_0)$ ($\lambda_0$ is non-vanishing, $X_0$ has at most $s$ nonzero rows) is identifiable up to the transformation group in \eqref{eq:10}, then $N\geq\frac{n-1}{n-s}$.
\end{proposition}
\begin{proof}
The matrix $X_0$ has at least $n-s$ zero rows. If we know the locations of $n-s$ zero rows, the problem becomes a BGPC problem with a subspace constraint. The columns of $AX_0$ reside in an $s$-dimensional subspace. If $N<\frac{n-1}{n-s}$, the pair $(\lambda_0,X_0)$ is not identifiable up to scaling and circular shift. The proof is almost identical to that of Proposition \ref{pro:necessary}.

The pair $(\lambda_0,X_0)$ cannot be identified even if we know the locations of $n-s$ zero rows. Hence it is not identifiable without knowing the locations of zero rows.
\end{proof}

The above necessary condition gives a tight lower bound on sample complexity. Morrison et al. \cite{Morrison2009} showed the same necessary condition for SAR autofocus (in the case of known row support of $X_0$). The two sample complexities, $N\geq s$, as is required by Theorem \ref{thm:jointsparsity} implicitly and Corollary \ref{cor:jointsparsity} explicitly, and $N\geq \frac{n-1}{n-s}$, coincide when $s=1$ or $s=n-1$. The gap between the sufficient condition and the necessary condition is analyzed next.

\subsubsection{Gap Between the Sufficient and the Necessary Conditions}\label{sec:gap2}
The sample complexity $N\geq s$ in the sufficient condition and the sample complexity $N\geq \frac{n-1}{n-s}$ in the necessary condition can be represented by the regions to the right of the line segment and the hyperbola (cf. Figure \ref{fig:ptjs}).
\begin{figure}[htbp]%
\centering
\input{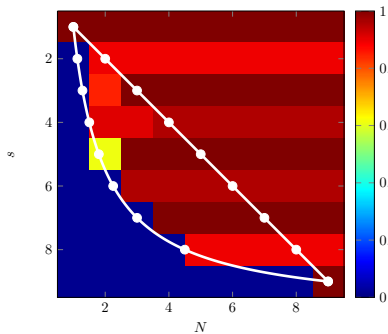}
\caption{The sample complexities for BGPC with DFT matrix and a joint sparsity constraint, and the ratio of identifiable pairs generated randomly.}%
\label{fig:ptjs}%
\end{figure}

Algorithm 2 can be used to check the identifiability of BGPC with DFT matrix and a joint-sparsity constraint. Given $Y=\diag(\lambda_0)FX_0$ that has no zero rows and joint support of $X_0$ that has cardinality $s$, Algorithm 2 returns a boolean value indicating whether or not $(\lambda_0,X_0)$ is identifiable up to the transformation group in \eqref{eq:10}. The procedure enumerates all joint supports of cardinality $s$.
\begin{algorithm}[H]
\caption{Identifiability of the BGPC problem with DFT matrix and a joint sparsity constraint}
\begin{algorithmic}
\STATE \textbf{input}: $Y$, the joint support $J$ \quad \textbf{output}: identifiability of $(\lambda_0,X_0)$
\FORALL{support $J'$ such that $|J'|=s$}
\STATE $G_{J'} \leftarrow \begin{bmatrix}[\diag(Y^{(:,1)})]^*F^{(:,J'^c)} & [\diag(Y^{(:,2)})]^*F^{(:,J'^c)} & \cdots & [\diag(Y^{(:,N)})]^*F^{(:,J'^c)} \end{bmatrix}^*$
\IF{$\rank(G_{J'})\leq n-2$}
\RETURN False
\ENDIF
\IF{$\rank(G_{J'})= n-1$ \AND $J'$ is not a shifted version of $J$}
\RETURN False
\ENDIF
\ENDFOR
\RETURN True
\end{algorithmic}
\end{algorithm}

\begin{proposition}
Given $Y=\diag(\lambda_0)FX_0$ that has no zero rows and the joint support of $X_0$ that has cardinality $s$, the pair $(\lambda_0,X_0)$ is identifiable (up to the transformation group in \eqref{eq:10}) if Algorithm 2 returns True, and not identifiable otherwise.
\end{proposition}
\begin{proof}
The matrix $G_{J'}\in\bbC^{N(n-s)\times n}$ satisfies that $G_{J'}x=\operatorname{vec}(F^{(:,J'^c)*}\diag(x)Y)$, where $F^{(:,J'^c)*}$ is an annihilator of the column space of $F^{(:,J')}$.
Given $Y$ that has no zero rows, any solution to the BGPC problem $(\lambda,X)$ satisfies that $\lambda$ is non-vanishing, and that the entrywise inverse of $\lambda$ is a null vector of $G_{J'}$, where $J'$ is the joint support of $X$. On the other hand, any null vector of $G_{J'}$ produces a solution $(\lambda,X)$, where $X$ is supported on $J'$.

If Algorithm 2 returns False, then at least one of the following two cases happens:
\begin{enumerate}
	\item $\rank(G_{J'})\leq n-2$ for some $|J'|=s$. By the proof of Proposition \ref{pro:necessary}, the solution is not identifiable even if the support $J'$ is known.
	\item $\rank(G_{J'})= n-1$ for some $J'$ that is not a shifted version of $J$. There exists a solution $(\lambda,X)$, for which $X\notin [X_0]^R_\scrT$. Therefore $(\lambda_0,X_0)$ is not identifiable.
\end{enumerate}
In either case, $(\lambda_0,X_0)$ is not identifiable up to the transformation group in \eqref{eq:10}.

If Algorithm 2 returns True, then $\rank(G_{J'})\geq n-1$ for all $J'$ of cardinality $s$, and $\rank(G_{J'})= n-1$ only if $J'$ is a shifted version of  $J$. Hence any solution $(\lambda,X)$ must satisfy that the joint support $J'$ is a shifted version of $J$. Now, given any shifted joint support $J'$, there exists a solution $(\lambda_{J'},X_{J'})\in[(\lambda_0,X_0)]_\scrT$. Therefore $G_{J'}$ has at least one null vector $x_{J'}$, which is the entrywise inverse of $\lambda_{J'}$. Hence $\rank(G_{J'})= n-1$, and the null vectors of $G_{J'}$ reside in the one-dimensional subspace spanned by $x_{J'}$. It follows that given the joint support $J'$, $\lambda$ in any solution must be a scaled version of $\lambda_{J'}$. Therefore $\lambda\in[\lambda_{J'}]^L_\scrT=[\lambda_0]^L_\scrT$. On the other hand, given non-vanishing $\lambda_0$, $\diag(\lambda_0)F$ has full rank and the recovery of $X_0$ has to be unique. Hence, by Corollary \ref{cor:ibip}, $(\lambda_0,X_0)$ is identifiable up to the transformation group in \eqref{eq:10}.
\end{proof}

The sufficient condition in Theorem \ref{thm:jointsparsity} is not necessary, as shown by the following numerically constructed counter-examples.
Let $n=10$, $1\leq s\leq 9$, and $1\leq N\leq 9$. The joint support $J$ of the columns of $X_0\in\bbR^{n\times N}$ is chosen uniformly at random. The entries of $\lambda_0\in\bbR^{n}$ and the nonzero entries of $X_0$ are generated as iid Gaussian random variables $N(0,1)$. We use Algorithm 2 to determine whether $(\lambda_0,X_0)$ is identifiable up to the transformation group in \eqref{eq:10}. For every value of $s$ and $N$, and every support $J$ of cardinality $s$, the numerical experiment is repeated independently. The ratio of identifiable pairs as a function of $(s,N)$ is shown in Figure \ref{fig:ptjs}. When $\frac{n-1}{n-s}\leq N<s$ (between the line and the hyperbola), the ratio of identifiable pairs is nonzero. Therefore, $N\geq s$ is not necessary.

The necessary condition in Proposition \ref{pro:necessary2} is not sufficient. This too can be demonstrated by Figure \ref{fig:ptjs}. The ratio of identifiable pairs is less than 1 in some regions to the right of the hyperbola. Unidentifiable examples of $(\lambda_0,X_0)$ that satisfy the necessary condition can be found in Appendix \ref{app:necens}.

As shown by Figure \ref{fig:ptjs}, when $N<\frac{n-1}{n-s}$ (to the left of the hyperpola), the pairs are not identifiable. When $N\geq s$ (to the right of the line segment), the identifiability hinges on the joint support of the columns of $X_0$. Most supports are not periodic, hence most pairs are identifiable. When $\frac{n-1}{n-s}\leq N<s$ (between the line and the hyperbola), the situation is more complicated. Besides periodic supports, other joint supports of $X_0$ can also cause non-identifiability. However, given some ``good'' joint support of $X_0$ that depends on both $s$ and $N$, a randomly chosen $(\lambda_0,X_0)$ is identifiable almost surely. Recall that non-periodicity of the joint support is necessary, hence ``good'' supports are a subset of non-periodic supports when $\frac{n-1}{n-s}\leq N< s$. For example, when $s=5$ and $N=2$, about 60\% of the non-periodic supports are ``good''. When $s=7$ and $N=3$, there is no ``good'' support. When $s=7$ and $N=4$, all non-periodic supports are ``good''. We have the following conjecture:
\begin{conjecture}
In the BGPC problem with DFT matrix and a joint sparsity constraint, if $N\geq \frac{n-1}{n-s}$, then for almost all $\lambda_0\in\bbC^n$ and almost all $X_0\in\bbC^{n\times N}$ that has $s$ nonzero rows and some ``good'' joint support, the pair $(\lambda_0,X_0)$ is identifiable up to the transformation group $\scrT$ defined in \eqref{eq:10}.
\end{conjecture}

\subsubsection{Extensions of the Model}

The results in Section \ref{sec:js} apply to $A=F$. This corresponds to MBD where the multiple channels are jointly sparse in the standard basis. Since the product of two circulant matrices is still a circulant matrix, we can easily show that the above results also apply to $A=FC$, where $C$ is a known invertible circulant matrix. This corresponds to MBD where the multiple channels are jointly sparse in the basis formed by the columns of $C$. In fact, results such as Theorem \ref{thm:jointsparsity} can also be derived for other matrices. In Section \ref{sec:piecewiseconstant}, we derive a sufficient condition for the identifiability of piecewise constant signals.

Although the results in Section \ref{sec:js} deal with 1D circular convolutions, extensions to higher-dimensional circular convolutions are straightforward. Let us consider a 2D MBD problem with a joint sparsity constraint as an example, and present a sufficient condition analogous to Theorem \ref{thm:jointsparsity}. Here $A=F\otimes F\in\bbC^{n\times n}$ is the 2D DFT matrix, where $F\in\bbC^{\sqrt{n}\times\sqrt{n}}$ is the 1D DFT matrix. In the 2D problem, the row index of $X$ can be represented by a pair of vertical and horizontal indices. For example, the $j$-th row of $X$ corresponds to the following index pair:
\[
(j^v,j^h)=\biggl( j-\sqrt{n}\bigl\lfloor\frac{j-1}{\sqrt{n}}\bigr\rfloor, \bigl\lfloor\frac{j-1}{\sqrt{n}}\bigr\rfloor+1 \biggr),
\]
where $\lfloor \cdot \rfloor$ denotes the floor operation. Repeating the procedure in Section \ref{sec:equivalence}, the transformation group for the 2D problem is defined by: 
\begin{equation}
\Gamma(F\otimes F) = \bigl\{\gamma=\sigma\sqrt{n} (F\otimes F)^{(:,k)}: \sigma\in\bbC \text{ is nonzero}, k\in\{1,2,\cdots,n\}\bigr\}.
\label{eq:25}
\end{equation}
\begin{equation}
\scrT = \{\calT:\calT(\lambda,X)=(\lambda./\gamma,(F\otimes F)^{*}\diag(\gamma)(F\otimes F)X) \text{ for some $\gamma\in \Gamma(F\otimes F)$}\}.
\label{eq:26}
\end{equation}
An equivalence transformation $\calT\in\scrT$ maps $X$ into a scaled 2D circular shift version of itself. The periodicity is defined as follows:
\begin{definition}
The index set $J=\{(j^v_1,j^h_1),(j^v_2,j^h_2),\cdots,(j^v_s,j^h_s)\}\subset\{1,2,\cdots,\sqrt{n}\}^2$ is said to be periodic with period $(\ell^v,\ell^h)$ ($\ell^v$ and $\ell^h$ being integers such that $0\leq\ell^v,\ell^h<\sqrt{n}$ and at least one of the two integers is nonzero), if $J=\{(j^v_1+\ell^v,j^h_1+\ell^h),(j^v_2+\ell^v,j^h_2+\ell^h),\cdots,(j^v_s+\ell^v,j^h_s+\ell^h)\}$ (modulo $(\sqrt{n},\sqrt{n})$).
\end{definition}
For example, if $\sqrt{n}=6$, then the index set $\{(1,1),(1,4)\}$ is periodic with period $(0,3)$. The index set $\{(1,1),(4,4)\}$ is periodic with period $(3,3)$. The index set $\{(1,1),(4,1),(1,4),(4,4)\}$ is periodic with period $(3,0)$, $(0,3)$, or $(3,3)$. The index set $\{(1,1),(5,3),(3,5)\}$ is periodic with period $(4,2)$ or $(2,4)$. The last two examples are shown in Figure \ref{fig:2dperiodic}.
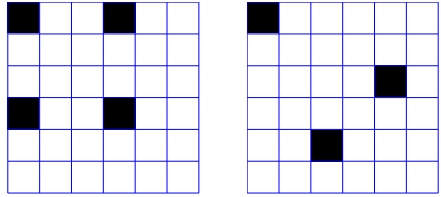
\begin{figure}[htbp]
\centering
%
%
\begin{tikzpicture}[scale=0.5]

\begin{axis}[%
width=2.5in,
height=2in,
axis on top,
scale only axis,
xmin=0.5,
xmax=8,
y dir=reverse,
ymin=0.5,
ymax=6.5,
hide axis,
name=plot1
]
\addplot [forget plot] graphics [xmin=0.5,xmax=6.5,ymin=0.5,ymax=6.5] {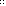};
\addplot [color=white,solid,forget plot]
  table[row sep=crcr]{%
6.5	0.5\\
8	0.5\\
};
\addplot [color=blue,solid,forget plot]
  table[row sep=crcr]{%
0.5	0.5\\
6.5	0.5\\
};
\addplot [color=blue,solid,forget plot]
  table[row sep=crcr]{%
0.5	0.5\\
0.5	6.5\\
};
\addplot [color=blue,solid,forget plot]
  table[row sep=crcr]{%
0.5	1.5\\
6.5	1.5\\
};
\addplot [color=blue,solid,forget plot]
  table[row sep=crcr]{%
1.5	0.5\\
1.5	6.5\\
};
\addplot [color=blue,solid,forget plot]
  table[row sep=crcr]{%
0.5	2.5\\
6.5	2.5\\
};
\addplot [color=blue,solid,forget plot]
  table[row sep=crcr]{%
2.5	0.5\\
2.5	6.5\\
};
\addplot [color=blue,solid,forget plot]
  table[row sep=crcr]{%
0.5	3.5\\
6.5	3.5\\
};
\addplot [color=blue,solid,forget plot]
  table[row sep=crcr]{%
3.5	0.5\\
3.5	6.5\\
};
\addplot [color=blue,solid,forget plot]
  table[row sep=crcr]{%
0.5	4.5\\
6.5	4.5\\
};
\addplot [color=blue,solid,forget plot]
  table[row sep=crcr]{%
4.5	0.5\\
4.5	6.5\\
};
\addplot [color=blue,solid,forget plot]
  table[row sep=crcr]{%
0.5	5.5\\
6.5	5.5\\
};
\addplot [color=blue,solid,forget plot]
  table[row sep=crcr]{%
5.5	0.5\\
5.5	6.5\\
};
\addplot [color=blue,solid,forget plot]
  table[row sep=crcr]{%
0.5	6.5\\
6.5	6.5\\
};
\addplot [color=blue,solid,forget plot]
  table[row sep=crcr]{%
6.5	0.5\\
6.5	6.5\\
};
\end{axis}

\begin{axis}[%
width=2in,
height=2in,
axis on top,
scale only axis,
xmin=0.5,
xmax=6.5,
y dir=reverse,
ymin=0.5,
ymax=6.5,
hide axis,
at=(plot1.right of south east),
anchor=left of south west
]
\addplot [forget plot] graphics [xmin=0.5,xmax=6.5,ymin=0.5,ymax=6.5] {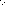};
\addplot [color=blue,solid,forget plot]
  table[row sep=crcr]{%
0.5	0.5\\
6.5	0.5\\
};
\addplot [color=blue,solid,forget plot]
  table[row sep=crcr]{%
0.5	0.5\\
0.5	6.5\\
};
\addplot [color=blue,solid,forget plot]
  table[row sep=crcr]{%
0.5	1.5\\
6.5	1.5\\
};
\addplot [color=blue,solid,forget plot]
  table[row sep=crcr]{%
1.5	0.5\\
1.5	6.5\\
};
\addplot [color=blue,solid,forget plot]
  table[row sep=crcr]{%
0.5	2.5\\
6.5	2.5\\
};
\addplot [color=blue,solid,forget plot]
  table[row sep=crcr]{%
2.5	0.5\\
2.5	6.5\\
};
\addplot [color=blue,solid,forget plot]
  table[row sep=crcr]{%
0.5	3.5\\
6.5	3.5\\
};
\addplot [color=blue,solid,forget plot]
  table[row sep=crcr]{%
3.5	0.5\\
3.5	6.5\\
};
\addplot [color=blue,solid,forget plot]
  table[row sep=crcr]{%
0.5	4.5\\
6.5	4.5\\
};
\addplot [color=blue,solid,forget plot]
  table[row sep=crcr]{%
4.5	0.5\\
4.5	6.5\\
};
\addplot [color=blue,solid,forget plot]
  table[row sep=crcr]{%
0.5	5.5\\
6.5	5.5\\
};
\addplot [color=blue,solid,forget plot]
  table[row sep=crcr]{%
5.5	0.5\\
5.5	6.5\\
};
\addplot [color=blue,solid,forget plot]
  table[row sep=crcr]{%
0.5	6.5\\
6.5	6.5\\
};
\addplot [color=blue,solid,forget plot]
  table[row sep=crcr]{%
6.5	0.5\\
6.5	6.5\\
};
\end{axis}
\end{tikzpicture}%
\caption{Examples of 2D periodic index sets.}%
\label{fig:2dperiodic}%
\end{figure}

Here is the sufficient condition for the 2D problem, whose proof is almost identical to that of Theorem \ref{thm:jointsparsity}.
\begin{theorem}\label{thm:2djointsparsity}
In the BGPC problem with 2D DFT matrix $F\otimes F\in\bbC^n$ and a joint sparsity constraint at sparsity level $s$, the pair $(\lambda_0,X_0)\in \bbC^n\times \Omega_\calX$ is identifiable up to the transformation group $\scrT$ defined in \eqref{eq:26} if the following conditions are met:
\begin{enumerate}
	\item Vector $\lambda_0$ is non-vanishing.
	\item Matrix $X_0$ has exactly $s$ nonzero rows and rank $s$.
	\item The joint support of the columns of $X_0$, represented in the index pair form, is not periodic.
\end{enumerate}
\end{theorem}


\subsection{Identifiability of Piecewise Constant Signals}\label{sec:piecewiseconstant}
Define the finite difference matrix $D\in\bbC^{n\times n}$ and its inverse as:
\[
D=
\begin{bmatrix}
1 & & & \\
-1 & 1 & & \\
& \ddots & \ddots & \\
& & -1 & 1
\end{bmatrix},
\quad
D^{-1}=
\begin{bmatrix}
1 & & &\\
1 & 1 & &\\
\vdots & \vdots & \ddots &\\
1 & 1 & \cdots & 1
\end{bmatrix}.
\]
A piecewise constant signal $u$ can be sparsified by the finite difference operator $D$. Equivalently, $u$ has the representation $u=D^{-1}x$ in which $x$ is sparse. If $U=D^{-1}X$ in which the columns of  $X$ are jointly sparse, then the columns of $U$ are piecewise constant and the discontinuities are at the same locations.

In this section, we consider the following blind deconvolution problem. The observation model is $Y=\diag(\lambda_0)FD^{-1}X_0$, where the matrix $X_0$ has at most $s$ nonzero rows. The non-vanishing vector $\lambda_0$ is the DFT of the filter. The columns of $D^{-1}X_0$ are the signals, which are piecewise constant and share the same discontinuities. An example is deblurring of hyperspectral images. The recovery of $(\lambda_0,X_0)$ is the BGPC problem with $A=FD^{-1}$ and a joint sparsity constraint.

First, we need to figure out the ambiguity transformation group. The structured matrix $P = A^{-1}\diag(\gamma)A= DF^*\diag(\gamma)FD^{-1} = DCD^{-1}$ is
\begin{equation}
P = 
\begin{bmatrix}
\sum_{j=1}^{n}c^{(j)} & \sum_{j=2}^{n}c^{(j)} & \sum_{j=2}^{n-1}c^{(j)} & \sum_{j=2}^{n-2}c^{(j)} & \cdots & c^{(2)}\\
0 & c^{(1)}-c^{(2)} & c^{(n)}-c^{(2)} & c^{(n-1)}-c^{(2)} & \cdots & c^{(3)}-c^{(2)} \\
0 & c^{(2)}-c^{(3)} & c^{(1)}-c^{(3)} & c^{(n)}-c^{(3)} & \cdots & c^{(4)}-c^{(3)} \\
0 & c^{(3)}-c^{(4)} & c^{(2)}-c^{(4)} & c^{(1)}-c^{(4)} & \cdots & c^{(5)}-c^{(4)} \\
\vdots & \vdots & \vdots & \vdots & \ddots & \vdots \\
0 & c^{(n-1)}-c^{(n)} & c^{(n-2)}-c^{(n)} & c^{(n-3)}-c^{(n)} & \cdots & c^{(1)}-c^{(n)}
\end{bmatrix},
\label{eq:13}
\end{equation}
where $C=F^*\diag(\gamma)F$ is a circulant matrix whose first column is
\[
\frac{1}{\sqrt{n}}F^*\gamma=c=[c^{(1)},c^{(2)},\cdots,c^{(n)}]\transpose.
\]
For $P$ to be a generalized permutation matrix, we must have $c^{(2)}=c^{(3)}=\cdots=c^{(n)}=0$, and $c^{(1)}\neq 0$. Hence $\gamma = \sqrt{n}Fc = c^{(1)}[1,1,\cdots,1]\transpose$. The ambiguity transformation group in \eqref{eq:8} becomes \eqref{eq:14}. We only allow an unknown scaling in the recovery.

Next we investigate identifiability up to scaling within the framework of Section \ref{sec:bip} and derive a sufficient condition. As in Theorem \ref{thm:jointsparsity}, one of the requirements is in terms of the joint support of the columns of $X_0$. We need the following definitions to state this sufficient condition.
\begin{definition}
Let the index sets $J_1,J_2,\cdots,J_T$ be the nodes of an undirected graph. There is an edge between $J_{t_1}$ and $J_{t_2}$ ($1\leq t_1<t_2\leq T$) if $J_{t_1}\bigcap J_{t_2}\neq\emptyset$. The index sets $J_1,J_2,\cdots,J_T$ are said to be connected if the above graph is connected.
\end{definition}
\begin{definition}\label{def:friendly}
The index set $J=\{j_1,j_2,\cdots,j_s\}\subset\{1,2,\cdots,n\}$ is said to be ``friendly'' if for any $0\leq k_1<k_2<\cdots<k_{n-s}\leq n-1$, the circularly shifted index sets $J_1,J_2,\cdots,J_{n-s}$, defined by  $J_t = \{j_1+k_t,j_2+k_t,\cdots,j_s+k_t\}$ (modulo $n$), satisfy that
\begin{enumerate}
	\item $|\bigcup_{t=1}^{n-s}J_t|\geq n-1$.
	\item $J_1,J_2,\cdots,J_{n-s}$ are connected.
\end{enumerate}
We make the convention that $\{1,2,\cdots,n\}$ is friendly.
\end{definition}

If the index set $J$ is friendly, and the entries indexed by its circularly shifted version $J_t$ ($1\leq t\leq n-s$) are equivalent in some sense, then due to transitivity of the equivalence relation, and the connectivity of the circularly shifted index sets, at least $n-1$ out of $n$ entries are equivalent. This property is used in the proof of Theorem \ref{thm:piecewiseconstant}.

\begin{remark}
If the index set $J$ is friendly, then its flipped and shifted versions are also friendly.
\end{remark}

We have the following propositions regarding the ``friendliness'' of an index set. Proposition \ref{pro:friendly0} shows that, for a non-trivial problem, a friendly index set must have cardinality at least $3$, which helps to avoid degeneracy in the proof of Theorem \ref{thm:piecewiseconstant}. Propositions \ref{pro:friendly1} and \ref{pro:friendly2} give two sufficient conditions for friendliness, which makes the property more readily interpretable. Corollary \ref{cor:cover} gives an alternative characterization of Condition 1 in Definition \ref{def:friendly}. See Appendix \ref{app:prooffriendly} for the proofs.
\begin{proposition}\label{pro:friendly0}
If $n\geq 4$ and the index set $J$ is friendly, then $|J|\geq 3$.
\end{proposition}
\begin{proposition}\label{pro:friendly1}
The index set $J$ is friendly if $|J|\geq 3$ and $J$ is contiguous\footnote{Index sets like $\{n,1,2\}$ are considered contiguous due to the circularity.}.
\end{proposition}

\begin{proposition}\label{pro:friendly2}
The index set $J$ is friendly if $|J|>\frac{n}{2}$ and $J$ is not periodic.
\end{proposition}

\begin{corollary}\label{cor:cover}
Let $|J|=s<n$. Then $|\bigcup_{t=1}^{n-s}J_t|\geq n-1$ for all choices of $n-s$ shifted index sets $J_t$ if and only if $J$ is not periodic.
\end{corollary}

Here is the sufficient condition for identifiability of piecewise constant signals.
\begin{theorem}\label{thm:piecewiseconstant}
Consider the BGPC problem with $A=FD^{-1}$ and two constraints: $\lambda$ is non-vanishing, and the columns of $X$ are jointly $s$-sparse. The pair $(\lambda_0,X_0)\in \bbC^n\times \Omega_\calX$ is identifiable up to an unknown scaling, if the following conditions are met (assume that $n\geq 4$ and $J=\{j_1,j_2,\cdots,j_{s}\}$ denotes the joint support of the columns of $X_0$):
\begin{enumerate}
	\item The vector $\lambda_0$ is non-vanishing.
	\item The matrix $X_0$ has exactly $s$ nonzero rows, and has rank $s$.
	\item $1\notin J$.
	\item $\{1\}\bigcup J$ is friendly.
\end{enumerate}
\end{theorem}
\begin{proof}
First, given non-vanishing $\lambda_0$ and $A=FD^{-1}$, the matrix $\diag(\lambda_0)FD^{-1}$ has full rank. If $\diag(\lambda_0)FD^{-1}X_0 =\diag(\lambda_0)FD^{-1}X_1$, then $X_1=X_0$. Hence, given $\lambda_0$, the recovery of $X_0$ is unique. By Corollary \ref{cor:ibip}, to establish the result, we only need to show that $\lambda_0$ is identifiable up to an unknown scaling.

Assuming that Conditions 1-4 of the theorem are satisfied, we show that $\lambda_1$ is a scaled version of $\lambda_0$, if $\diag(\lambda_0)FD^{-1}X_0=\diag(\lambda_1)FD^{-1}X_1$ for $(\lambda_1,X_1)$ that satisfies the two constraints. The matrix $\diag(\lambda_0)FD^{-1}$ has full rank, hence both $X_0$ and $X_1$ have rank $s$. 
Without loss of generality, we may assume that $X_0$ and $X_1$ only have $s$ columns, which are linearly independent, by removing redundant columns at the same locations in both matrices. They both have full column rank $s$ and exactly $s$ nonzero rows. 
By assumption, the vectors $\lambda_0$ and $\lambda_1$ are non-vanishing. Write $X_1$ in terms of $X_0$, $X_1=PX_0$, where
\[
P=DF^*[\diag(\lambda_1)]^{-1}\diag(\lambda_0)FD^{-1}=DF^*\diag(\gamma)FD^{-1}.
\]
The matrix $P$ has the structure in \eqref{eq:13} where $c=\frac{1}{\sqrt{n}}F^*\gamma=\frac{1}{\sqrt{n}}F^*(\lambda_0./\lambda_1)$. Furthermore, $P$ satisfies \eqref{eq:12} in the proof of Theorem \ref{thm:jointsparsity}. The submatrix $P^{(:,J)}=X_1X_0^{\dagger(:,J)}$ has at most $s$ nonzero rows and at least $n-s$ zero rows. The submatrix $P^{(2:n,J)}$ has at least $n-s-1$ zero rows. We denote the corresponding index set by $K = \{k_1,k_2,\cdots,k_{n-s-1}\}$. By \eqref{eq:13}, the row $P^{(k,J)}$ ($k\in K$) is:
\[
P^{(k,J)}=\left[c^{(k+1-j_1)}-c^{(k)}, c^{(k+1-j_2)}-c^{(k)}, \cdots , c^{(k+1-j_s)}-c^{(k)}\right].
\]
The index set $J_k=\{k,k+1-j_1,k+1-j_2,\cdots,k+1-j_s\}$ is a flipped and shifted version of $\{1\}\bigcup J = \{1,j_1,j_2,\cdots,j_s\}$. The above row $P^{(k,J)}$ is zero, which means all the entries of the subvector $c^{(J_k)}$ are equal. By the assumption that $\{1\}\bigcup J$ is friendly, the index sets $J_{k_1},J_{k_2},\cdots,J_{k_{n-s-1}}$ are connected. That means all the entries of $c$ indexed by $\bigcup_{t=1}^{n-s-1}J_{k_t}$ are equal. Besides, $|\bigcup_{t=1}^{n-s-1}J_{k_t}|\geq n-1$. That means either all the entries of $c$ are equal or there is one entry with a different value. There are three different cases:
\begin{enumerate}
	\item All the entries of $c$ are equal. Then the vector $\lambda_0./\lambda_1 = \sqrt{n}Fc$ has $n-1$ zeros, which contradicts the assumption that $\lambda_0,\lambda_1$ are non-vanishing.
	\item All but the $k_0$-th entry of $c$ are equal, where $k_0\neq 1$. Then all the entries of $P^{(2:n,J)}$ that do not contain $c^{(k_0)}$ are zeros, and all the entries that contain $c^{(k_0)}$ are nonzeros. The rows indexed by $K$ are zeros, hence they do not contain $c^{(k_0)}$. The row indexed by $k_0$ is shown in \eqref{eq:15}, and is nonzero. The rows that contain any of the $s$ entries in \eqref{eq:16} are also nonzeros. 
\begin{equation}
c^{(k_0-j_1+1)}-c^{(k_0)},c^{(k_0-j_2+1)}-c^{(k_0)},\cdots,c^{(k_0-j_s+1)}-c^{(k_0)}
\label{eq:15}
\end{equation}
\begin{equation}
c^{(k_0)}-c^{(k_0+j_1-1)},c^{(k_0)}-c^{(k_0+j_2-1)},\cdots,c^{(k_0)}-c^{(k_0+j_s-1)}
\label{eq:16}
\end{equation}
Note that no two entries in \eqref{eq:16} can belong to the same row; no entry in \eqref{eq:16} belongs to the row in \eqref{eq:15}. If every entry in \eqref{eq:16} belonged to a row in $P^{(2:n,J)}$, there would be $s+1$ nonzero rows in $P^{(2:n,J)}$. The number of nonzero rows in $P^{(:,J)}$ is at most $s$. Hence, one of the $s$ entries in \eqref{eq:16} is not in any row of $P^{(2:n,J)}$. By observation, the only entry that could be missing is $c^{(k_0)}-c^{(1)}$. Assume that, without loss of generality, $c^{(k_0)}-c^{(k_0+j_1-1)}$ is not in any row of $P^{(2:n,J)}$. That implies $k_0+j_1-1=1$ (modulo $n$). Hence there exists an entry in the first row $P^{(1,j_1)}=\sum_{j=2}^{n+2-j_1}{c^{(j)}}=\sum_{j=2}^{k_0}{c^{(j)}}$. Since $n\geq 4$ and $\{1\}\bigcup J$ is friendly, by Proposition \ref{pro:friendly0}, $|J|\geq 2$. Hence there exists another entry in the first row $P^{(1,j_2)}=\sum_{j=2}^{n+2-j_2}{c^{(j)}}$. Since there are $s$ nonzero rows in $P^{(2:n,J)}$, the first row $P^{(1,J)}$ must be zero. Hence,
\[
\sum_{j=2}^{k_0}{c^{(j)}}=\sum_{j=2}^{n+2-j_2}{c^{(j)}}=0.
\]
Recall that all the entries of $c$ are equal except for $c^{(k_0)}$. It follows that $c^{(1)}=c^{(2)}=\cdots = c^{(n)}=0$, resulting in a contradiction.
	\item All but the first entry of $c$ are equal. Then all the entries of $P^{(2:n,J)}$ that do not contain $c^{(1)}$ are zeros, and all the entries that contain $c^{(1)}$ are nonzeros. In particular, the entries $c^{(1)}-c^{(j_1)},c^{(1)}-c^{(j_2)},\cdots,c^{(1)}-c^{(j_s)}$ in the rows indexed by $j_1,j_2,\cdots,j_s$ are nonzeros. Hence the first row $P^{(1,J)}$ must be zero. Therefore, $c^{(2)}=c^{(3)}=\cdots=c^{(n)}=0$, and $c^{(1)}\neq 0$.
\end{enumerate}
The only case that does not cause a contradiction is the third, which leads to $c=[c^{(1)},0,0,\cdots,0]\transpose$ and $\lambda_0./\lambda_1=\sqrt{n}Fc=c^{(1)}[1,1,\cdots,1]\transpose$. Therefore, $\lambda_1=\frac{1}{c^{(1)}}\lambda_0$ is a scaled version of $\lambda_0$.
\end{proof}

A result for generic signals, analogous to Corollary \ref{cor:jointsparsity}, follows immediately.

The requirement $N \geq s$, implied by Theorem \ref{thm:piecewiseconstant}, is not necessary. We have the following necessary condition, which can be proved similarly to Proposition \ref{pro:necessary2}.
\begin{proposition}\label{pro:necessary3}
In the BGPC problem with $A=FD^{-1}$ and a joint sparsity constraint, if $(\lambda_0,X_0)$ ($\lambda_0$ is non-vanishing, $X_0$ has at most $s$ nonzero rows) is identifiable up to scaling, then $N\geq\frac{n-1}{n-s}$.
\end{proposition}

An analysis of the gap between the sufficient and the necessary conditions, similar to Section \ref{sec:gap2}, can be carried out for these results too. It is omitted for brevity.


\section{Universal Sufficient Condition for BGPC with a Sparsity Constraint}\label{sec:sparsity}
In this section, we consider the BGPC problem with a sparsity constraint on the total number of nonzero entries in the matrix $X$, denoted by $\norm{X}_0$.
Consider the following problem:
\begin{align*}
\text{($\mathrm{P3}$)}\quad\text{find}~~&(\lambda,X),\\
\text{s.t.}~~&\diag(\lambda)AX = Y,\\
& \lambda \in \bbC^n,~X \in \Omega_\calX=\{X\in\bbC^{n\times N}: \norm{X}_0\leq s\}.
\end{align*}
The measurement is $Y=\diag(\lambda_0)AX_0$. We only consider the case where $A\in\bbC^{n\times n}$ is an invertible square matrix. The vector $\lambda_0\in\bbC^n$ is non-vanishing. The matrix $X_0\in\bbC^{n\times N}$ has at most $s$ nonzero entries.

The ambiguity transformation group $\scrT$ associated with the matrix $A$ is the same as in Section \ref{sec:equivalence}. 
In Theorem \ref{thm:universal}, we show that $X_0$ is identifiable up to a generalized permutation in the ambiguity transformation group associated with $A$ if the rows of $X_0$ form the most sparse basis of its row space. This is a universal sufficient condition for BGPC with a sparsity constraint, which applies to every invertible square matrix $A$. This universal result is derived using the general framework in Section \ref{sec:bip}.

\begin{theorem}\label{thm:universal}
In the BGPC problem with a sparsity constraint at sparsity level $s$, the pair $(\lambda_0,X_0)$ is identifiable up to the ambiguity transformation group $\scrT$ associated with $A$, if the following conditions are met:
\begin{enumerate}
	\item Vector $\lambda_0$ is non-vanishing.
	\item If an invertible matrix $P\in\bbC^{n\times n}$ satisfies that $\norm{PX_0}_0 \leq \norm{X_0}_0$, then $P$ is a generalized permutation matrix.
	\item $\norm{X_0}_0= s$.
\end{enumerate}
\end{theorem}

\begin{proof}
Given non-vanishing $\lambda_0$ and invertible $A$, the matrix $\diag(\lambda_0)A$ is invertible. Hence given $\lambda_0$, the matrix $X_0$ is identifiable. By Corollary \ref{cor:ibip}, we only need to show that $\lambda_0$ is identifiable. Suppose that $\diag(\lambda_0)AX_0=\diag(\lambda_1)AX_1$ and $\norm{X_1}_0\leq s =\norm{X_0}_0$. By the above Condition 2, $X_0$ has full row rank $n$. Otherwise, there exists an invertible matrix $P$ that is not a permutation matrix and satisfies $PX_0=X_0$, which clearly violates Condition 2. The matrix $\diag(\lambda_0)A$ is invertible, hence $\rank(X_1)=\rank(X_0)=n$. There are no zero rows in $AX_0$ or $AX_1$. Hence $\lambda_1$ is also non-vanishing. Write $X_1$ in terms of $X_0$, $X_1=PX_0$, where $P=A^{-1}[\diag(\lambda_1)]^{-1}\diag(\lambda_0)A$. By the above Condition 2, $P$ has to be a generalized permutation matrix. By \eqref{eq:7} and \eqref{eq:8}, $\gamma = \lambda_0./\lambda_1\in \Gamma(A)$ and $\lambda_1\in[\lambda_0]^L_\scrT$. Therefore, $\lambda_0$ is identifiable.
\end{proof}

If the sparsity level is not known a priori, we can solve the following optimization problem ($\mathrm{P4}$). Under the above Conditions 1 and 2, the minimizer in ($\mathrm{P4}$) is unique up to the same transformation group. If the minimizer to ($\mathrm{P4}$) has sparsity $s$, then it is the solution to ($\mathrm{P3}$) as well.
\begin{align*}
\text{($\mathrm{P4}$)}\quad\underset{(\lambda,X)}{\text{min.}}~~&\norm{X}_0,\\
\text{s.t.}~~&\diag(\lambda)AX = Y,\\
& \lambda \in \bbC^n,~X \in \bbC^{n\times N}.
\end{align*}
The following universal sufficient condition follows by combining Theorem \ref{thm:universal} with results about the distribution of non-zero elements in random matrices and in the products of such matrices with vectors \cite{Spielman2013}.

\begin{theorem}
Suppose that the vector $\lambda_0$ is non-vanishing, the matrix $X_0\in\bbC^{n\times N}$ is Bernoulli-Gaussian random matrix, where $X_0=B\odot G$, the entries of $B$ are iid Bernoulli random variables $B(1,\theta)$, and the entries of $G$ are iid Gaussian random variables $N(0,1)$. If $\frac{1}{n}<\theta<\frac{1}{4}$ and $N>Cn\log n$ for a sufficiently large absolute constant C, then the pair $(\lambda_0,X_0)$ is identifiable in ($\mathrm{P4}$), up to the ambiguity transformation group $\scrT$ associated with $A$, with probability at least $1-\exp(-c\theta N)$ for some absolute constant $c$.
\end{theorem}

\begin{proof}
We prove the identifiability by showing that Condition 2 in Theorem \ref{thm:universal} is satisfied with probability at least $1-\exp(-c\theta N)$ given the above Bernoulli-Gaussian model. Assume that $P\in\bbC^{n\times n}$ is an invertible matrix but not a generalized permutation matrix. Since $P$ is invertible, there exists a permutation of $1,2,\cdots,n$, denoted by $j_1,j_2,\cdots,j_n$, such that the support of the $k$th row $P^{(k,:)}$ contain the index $j_k$, i.e., $P^{(k,j_k)}\neq 0$, for $1\leq k\leq n$. Since $P$ is not a generalized permutation matrix, there exists at least one row with more than one nonzero entries. If the row $P^{(k,:)}$ has only one nonzero entry $P^{(k,j_k)}$, then $\norm{(PX_0)^{(k,:)}}_0=\norm{P^{(k,:)}X_0}_0=\norm{X_0^{(j_k,:)}}_0$. Next, we show that if $P^{(k,:)}$ has more than one nonzero entries, then $\norm{(PX_0)^{(k,:)}}_0>\norm{X_0^{(j_k,:)}}_0$ with high probability.

By Lemma 17 in \cite{Spielman2013}, if the Bernoulli-Gaussian matrix $X_0$ satisfies that $\frac{1}{n}<\theta<\frac{1}{4}$ and $N>Cn\log n$ for a sufficiently large constant $C$, then the probability that there exists a vector $v\in \bbC^n$ with more than one nonzero entries such that $\norm{v^*X_0}_0\leq \frac{11}{9}\theta N$ is at most $\exp(-c_1\theta N)$, for some absolute constant $c_1$. Therefore, with probability at least $1-\exp(-c_1\theta N)$, 
\begin{equation}
\norm{(PX_0)^{(k,:)}}_0>\frac{11}{9}\theta N
\label{eq:17}
\end{equation} 
for every index $k$ such that $P^{(k,:)}$ has more than one nonzero entries.

By Lemma 18 in \cite{Spielman2013}, the probability that any row of the Bernoulli-Gaussian matrix $X_0$ has more than $\frac{10}{9}\theta N$ nonzero entries is at most $n\exp(-\theta N/243)$. Since $N>Cn\log n$ for a sufficiently large constant $C$, the probability $n\exp(-\theta N/243)\leq \exp(-c_2\theta N)$ for some absolute constant $c_2$. Therefore, with probability at least $1-\exp(-c_2\theta N)$, 
\begin{equation}
\norm{X_0^{(j_k,:)}}_0 \leq \frac{10}{9}\theta N
\label{eq:18}
\end{equation}
for every $k$.

Combining \eqref{eq:17} and \eqref{eq:18}, $\norm{(PX_0)^{(k,:)}}_0>\norm{X_0^{(j_k,:)}}_0$ for every index $k$ such that $P^{(k,:)}$ has more than one nonzero entries, with probability at least $1-\exp(-c\theta N)$ for some absolute constant $c$. Therefore, with the same probability,
\[
\norm{PX_0}_0=\sum_{k=1}^{n}\norm{(PX_0)^{(k,:)}}_0 >\sum_{k=1}^{n}\norm{X_0^{(j_k,:)}}_0 =\norm{X_0}_0.
\]
Equivalently, Condition 2 in Theorem \ref{thm:universal} is satisfied with probability at least $1-\exp(-c\theta N)$. 
\end{proof}


\section{Discussion} \label{sec:conclusions}
We defined identifiability of a bilinear inverse problem up to transformation groups. A general framework for proving identifiability was proposed. The framework was applied to the problem of BGPC. We showed sufficient conditions for unique recovery up to a transformation group under three scenarios, with a subspace constraint, with a joint sparsity constraint, and with a sparsity constraint, respectively. We also provided necessary conditions for the scenarios with a subspace constraint or a joint sparsity constraint. We developed a procedure to determine the ambiguity transformation groups for BGPC with joint sparsity or with sparsity constraints. We also designed algorithms that can check the identifiability for BGPC with subspace or with joint sparsity constraints, and demonstrated the tightness of our sample complexity bounds by numerical experiments.

The analysis in this paper is not always optimal. In certain cases, there exist gaps between the sufficient conditions and the necessary conditions. For example, in the scenario with DFT matrix and a joint sparsity constraint, the gap between the sample complexities in the sufficient and the necessary conditions is $N\geq s$ versus $N\geq \frac{n-1}{n-s}$. However, we believe that it would be possible to bridge these gaps by introducing more stringent assumptions (e.g., generic vectors and matrices).

One goal of this paper is to motivate more research into the identifiability of bilinear inverse problems. For BGPC, additional identifiability results can be obtained for different bases $A$ and different constraint sets $\Omega_\Lambda,\Omega_\calX$. For example, exploiting the extra information regarding $\lambda$ (positivity in inverse rendering, unit-modulus entries in SAR autofocus), is expected to provide less demanding conditions for identifiability. The merit of the framework in this paper for identifiability in bilinear inverse problems is not restricted to the demonstrated exemplary applications. It will be useful for analyzing a wider class of practical applications, including blind deconvolution (with a single channel and/or the linear convolution model), phase retrieval, dictionary learning, etc.

\appendix
\section{Example of a Non-trivial Annihilator}\label{app:nta}
Most bilinear mappings that arise in applications do not have non-trivial left or right annihilators, however this is not universally true. Here is an example in which the bilinear mapping does have a non-trivial right annihilator.
Assume that $z=x_0y_0^{(1)}\in \bbC^2$ in the following BIP:
\begin{align*}
\text{find}~~&(x,y),\\
\text{s.t.}~~&x y^{(1)} = z,\\
& x \in \bbC^2,~y \in \bbC^2.
\end{align*}
Then $(x_0,y_0)$ is identifiable up to the following transformation group:
\[
\scrT = \left\{\calT: \calT(x,y)=\left(\frac{1}{\sigma} x,[\sigma y^{(1)},y^{(2)}+\tau]\transpose\right)\text{ for some nonzero $\sigma\in\bbC$ and some $\tau\in\bbC$}\right\}.
\]
Let $\calT=(\calT_\calX,\calT_\calY)$, where $\calT_\calX(x)=\frac{1}{\sigma}x$, $\calT_\calY(y)=[\sigma y^{(1)},y^{(2)}+\tau]\transpose$. Note that $\calT_\calY$ is not a linear transformation if $\tau\neq 0$. In addition, Condition 2 in Corollary \ref{cor:ibip} is not necessary. Given $\calF(x_0,y_0)=\calF(x_0,y)$, i.e., $x_0y_0^{(1)}=x_0y^{(1)}$, it is not necessary that $y=y_0$. The reason is that the bilinear mapping $\calF$ has a non-trivial right annihilator $y=[0,1]\transpose$.

\section{Proof of Lemma \ref{lem:decomposable}}\label{app:proofdecomposable}
\begin{enumerate}
	\item If $A\in\bbC^{n\times m}$ has full row rank, then the rows of $A$ form a basis for $\calR(A)$ whose dimension is $n$. For every non-empty proper subset $J$ and its complement $J^c$, $\calR(A^{(J,:)})$ and $\calR(A^{(J^c,:)})$ are two subspaces whose dimensions are $|J|$ and $|J^c|$ respectively. Therefore,
\begin{align*}
\calR(A) &=\calR(A^{(J,:)})+\calR(A^{(J^c,:)}),\\
\dim(\calR(A)) &=n=|J|+|J^c|=\dim(\calR(A^{(J,:)}))+\dim(\calR(A^{(J^c,:)})).
\end{align*}
Therefore, the sum of two subspaces is a direct sum, and the row space of $A$ is decomposable.
	
	\item If the row space of $A$ is not decomposable, then $A$ does not have full row rank. If the matrix $A$ has full column rank, then $n\geq m$. 
	
Next, we prove $n>m$ by contradiction. Suppose that $n=m$. Since square matrix $A$ has full column rank, it must have full row rank, which causes a contradiction. Therefore, the assumption is false, and $n$ has to be greater than $m$.
	
	\item The row space of $A$ is not decomposable, if and only if the sum $\calR(A)=\calR(A^{(J,:)})+\calR(A^{(J^c,:)})$ is not a direct sum for any non-empty proper subset $J\subset\{1,2,\cdots,n\}$, or equivalently,
\[
\dim(\calR(A))<\dim(\calR(A^{(J,:)}))+\dim(\calR(A^{(J^c,:)})),
\]
for all non-empty proper subsets $J\subset\{1,2,\cdots,n\}$.
\end{enumerate}
\section{Examples of Ambiguity Transformation Groups}\label{app:tgforA}
In the BGPC problem with a joint sparsity constraint, the ambiguity transformation groups for $A$ can be figured out with the method in Section \ref{sec:equivalence}. The ambiguity transformation groups associated with $A=F$ and $A=FD^{-1}$ are shown in Section \ref{sec:equivalence} and Section \ref{sec:piecewiseconstant} respectively. We give more examples here.

The matrix $A$ introduces some ``mixing'' to the rows of $X$. If $A=I$, there is no mixing. The structured matrix $I^{-1}\diag(\gamma)I=\diag(\gamma)$ is a diagonal matrix. It is a generalized permutation matrix provided that $\gamma$ is non-vanishing. The set of $\gamma$ which produces a generalized permutation matrix is $\Gamma(I) = \{\gamma\in\bbC^{n}: \text{$\gamma$ is non-vanishing}\}$. The ambiguity transformation group is
\[
\scrT = \{\calT:\calT(\lambda,X)=(\lambda./\gamma,\diag(\gamma)X)\text{ for some non-vanishing $\gamma$}\}.
\]
In this case, any non-vanishing $\lambda$ is considered equivalent to $\lambda_0$. The identifiability of $(\lambda_0,X_0)$ with this transformation group is not an interesting problem.

For some $A$, the structured matrix $A^{-1}\diag(\gamma)A$ is already studied in the literature. For example, if $A$ is a DFT matrix, $A^{-1}\diag(\gamma)A$ is a circulant matrix. If $A$ is the discrete cosine transform (DCT) matrix, $A^{-1}\diag(\gamma)A$ is the sum of a symmetric Toeplitz matrix and a Hankel matrix \cite{Sanchez1995}. For other matrices, the structure of $A^{-1}\diag(\gamma)A$ can be figured out by symbolic computation. The matrix $A=FD^{-1}$ in Section \ref{sec:piecewiseconstant} is an example.
Another example is the Haar matrix $H_n$, corresponding to a wavelet transform.
The matrix $H_4$ and the structured matrix $H_4^{-1}\diag(\gamma)H_4$ are
\begin{equation*}
H_4 =\begin{bmatrix}1 & 1 & 1 & 1\\1 & 1 & -1& -1\\ 1 & -1 & 0 & 0\\ 0 & 0 & 1 & -1\end{bmatrix},
\end{equation*}
\begin{equation*}
H_4^{-1}\diag(\gamma)H_4 =\frac{1}{4}\begin{bmatrix}
\gamma^{(1)}+\gamma^{(2)}+2\gamma^{(3)} & \gamma^{(1)}+\gamma^{(2)}-2\gamma^{(3)} & \gamma^{(1)}-\gamma^{(2)} & \gamma^{(1)}-\gamma^{(2)}\\
\gamma^{(1)}+\gamma^{(2)}-2\gamma^{(3)} & \gamma^{(1)}+\gamma^{(2)}+2\gamma^{(3)} & \gamma^{(1)}-\gamma^{(2)} & \gamma^{(1)}-\gamma^{(2)}\\
\gamma^{(1)}-\gamma^{(2)} & \gamma^{(1)}-\gamma^{(2)} & \gamma^{(1)}+\gamma^{(2)}+2\gamma^{(4)} & \gamma^{(1)}+\gamma^{(2)}-2\gamma^{(4)}  \\
\gamma^{(1)}-\gamma^{(2)} & \gamma^{(1)}-\gamma^{(2)} & \gamma^{(1)}+\gamma^{(2)}-2\gamma^{(4)} & \gamma^{(1)}+\gamma^{(2)}+2\gamma^{(4)} 
\end{bmatrix}.
\end{equation*}
The structured matrix $H_4^{-1}\diag(\gamma)H_4$ is a generalized permutation matrix if and only if $\gamma^{(2)}=\gamma^{(1)}$, $\gamma^{(3)}=\pm\gamma^{(1)}$ and $\gamma^{(4)}=\pm\gamma^{(1)}$. The set $\Gamma(H_4)$ and the ambiguity transformation group $\scrT$ are
\[
\Gamma(H_4)=\{\gamma:\gamma^{(1)}=\gamma^{(2)}=\sigma, \gamma^{(3)}=\pm\sigma, \gamma^{(4)}=\pm\sigma, \text{ for some nonzero $\sigma\in\bbC$}\},
\]
\[
\scrT = \{\calT:\calT(\lambda,X)=(\lambda./\gamma,H_4^{-1}\diag(\gamma)H_4X) \text{ for some $\gamma\in \Gamma(H_4)$}\}.
\]

\section{Insufficiency of the Condition in Proposition \ref{pro:necessary2}}\label{app:necens}
The necessary condition in Proposition \ref{pro:necessary2} is not sufficient, even when the locations of the zero rows are known a priori. For example, when $n=7$, $s=4$, $\lambda_0\in\bbC^7$ is non-vanishing and
\[
X_0 = \begin{bmatrix}1&0&0\\0&1&0\\0&0&1\\0&0&0\\0&0&0\\0&0&0\\0&0&0\end{bmatrix},
\]
the pair $(\lambda_0,X_0)$ is not identifiable, even if we know that the last three rows of $X_0$ are zeros. There exists a cirulant matrix $P$ whose first column is $[1,2,0,0,0,0,0]\transpose$ such that
\[
X_1 = PX_0 = \begin{bmatrix}1&0&0\\2&1&0\\0&2&1\\0&0&2\\0&0&0\\0&0&0\\0&0&0\end{bmatrix},
\]
and $\lambda_1=\lambda_0./\gamma$, where $\gamma = \sqrt{n}F[1,2,0,0,0,0,0]\transpose$ is non-vanishing.

The above example is a degenerate case where the actual joint sparsity of $X_0$ is less than $s=4$. A non-degenerate $X_0$ may also not be identifiable, if there is no extra knowledge of the locations of the zero rows. For example,
\[
X_0 = \begin{bmatrix}1&3&2\\2&1&3\\3&2&1\\-29&-28.5&-17.5\\0&0&0\\0&0&0\\0&0&0\end{bmatrix}.
\]
There exists a circulant matrix $P$ whose first column is $[2,16,1,8,0.5,4,32]\transpose$, such that
\[
X_1=PX_0=\begin{bmatrix}63.5&31.75&95.25\\0&0&0\\-889&-889&-508\\0&0&0\\-444.5&-444.5&-254\\0&0&0\\-190.5&-127&-63.5\end{bmatrix},
\]
and $\lambda_1=\lambda_0./\gamma$, where $\gamma=\sqrt{n}F[2,16,1,8,0.5,4,32]\transpose$ is non-vanishing.

The above pathological examples reside in a set of measure zero. Next, we show that when $\rank(X_0)=s$ but the joint support of the columns of $X_0$ is periodic, the pair $(\lambda_0,X_0)$ is not identifiable. This set of unidentifiable $X_0$ has nonzero measure. Recall the proof of Theorem \ref{thm:jointsparsity}. Assume that the joint support of the columns of $X_0$ is periodic with period $\ell$. There exists a circulant matrix $P$ with two nonzero entries in the first column, indexed by $k_1$ and $k_2$, such that $k_2-k_1=\ell$ and $\gamma = \sqrt{n}FP^{(:,1)}$ is non-vanishing. Hence there exists $X_1=PX_0$ and $\lambda_1=\lambda_0./\gamma$ such that $\diag(\lambda_0)FX_0=\diag(\lambda_1)FX_1$ and $\lambda_1\notin[\lambda_0]^L_\scrT$. Therefore, $(\lambda_0,X_0)$ is not identifiable.

\section{Proof of the Propositions Regarding ``Friendliness''} \label{app:prooffriendly}
\begin{proof}[Proof of Proposition \ref{pro:friendly0}]
We prove by contraposition, i.e., if $n\geq 4$ and $|J|\leq 2$, then $J$ is not friendly.
First, if $J=\emptyset$ or $|J|=1$, then the circularly shifted index sets $J_1,J_2,\cdots,J_{n-s}$ are not connected.

Next, we show that if $n\geq 4$ and $|J|=2$, then $J_1,J_2,\cdots,J_{n-s}$ are not connected. Since all the circularly shifted index sets are equivalent, without loss  of generality, we may assume that $J=\{1,r\}$, where $2\leq r\leq \frac{n}{2}+1$. Then all the sets $\{r_1,r_2\}$ such that $r_1-r_2 = r-1$ (modulo $n$) or $r_2-r_1 = r-1$ (modulo $n$) are circularly shifted versions of $J$. There are a total of $n$ circularly shifted index sets.

If $r=\frac{n}{2}+1$, then $J$ is periodic. The sets like $\{r_1,r_1+\frac{n}{2}\}$ ($1\leq r_1\leq \frac{n}{2}$) are counted twice because $\frac{n}{2}=-\frac{n}{2}$ (modulo $n$). And these index sets are not connected.

If $n\geq 4$ and $r<\frac{n}{2}+1$, the $n$ index sets are $\{1,r\},\{2,r+1\},\cdots,\{n-r+1,n\},\{n-r+2,1\},\cdots,\{n,r-1\}$. By removing $\{r,2r-1\}$ and $\{n-r+2,1\}$, there are $n-2\geq 2$ index sets left. These circularly shifted versions of $J$ are not connected because $J=\{1,r\}$ is not connected to the rest.
\end{proof}

\begin{proof}[Proof of Proposition \ref{pro:friendly1}]
First, if $J$ is contiguous and $|J|=s$, then $n-s$ shifted contiguous index sets cover at least $s+(n-s-1)=n-1$ indices. Therefore, $|\bigcup_{t=1}^{n-s}J_t|\geq n-1$.

Next, we prove that the shifted index sets $J_1,J_2,\cdots,J_{n-s}$ are connected by showing that they form a cycle or a path in the graph. To this end, we show that between the $n-s$ pairs $(J_1,J_2),(J_2,J_3),\cdots,(J_{n-s},J_1)$, there are at least $n-s-1$ edges. Suppose the opposite, that there are fewer edges, for example two edges are missing in the above cyle. Then $n-s$ shifted contiguous index sets cover at least $s+s+(n-s-2)=n+s-2\geq n+1$ indices, a contradiction. 
\end{proof}

\begin{proof}[Proof of Proposition \ref{pro:friendly2}]
We first show that if $J$ is not periodic, then $|\bigcup_{t=1}^{n-s}J_t|\geq n-1$, or equivalently $|\bigcap_{t=1}^{n-s}J_t^c|\leq 1$. We prove the contrapositive, if there are two distinct indices $k',k''\in \bigcap_{t=1}^{n-s}J_t^c$ then $J$ is periodic. Note that $J_1^c,J_2^c,\cdots,J_{n-s}^c$ are all circularly shifted versions of the same index set $J^c=\{j^c_1,j^c_2,\cdots,j^c_{n-s}\}$.
\[
J^c=\{j^c_1,j^c_2,\cdots,j^c_{n-s}\}=\{k'-k_1,k'-k_2,\cdots,k'-k_{n-s}\}=\{k''-k_1,k''-k_2,\cdots,k''-k_{n-s}\} \quad \text{(modulo $n$)},
\]
Hence $J^c$ is periodic with period $\ell = |k''-k'|$, so is $J$.

Next we show that the shifted index sets are connected. If $|J|>\frac{n}{2}$, then $J_{t_1}\bigcap J_{t_2}\neq \emptyset$ for any $t_1,t_2$. There is an edge between every pair of nodes, hence the graph is a complete graph, which is connected.
\end{proof}

\begin{proof}[Proof of Corollary \ref{cor:cover}]
The sufficiency is shown in the proof of Proposition \ref{pro:friendly2}. 

Next we prove necessity. If $J$ is periodic with period $\ell$ and $|J|=s<n$, then for any $k',k''$ such that $k''-k'=\ell$, we can always apply the proper shifts $k_1=k'-j_1^c,k_2=k'-j_2^c,\cdots,k_{n-s}=k'-j_{n-s}^c$ such that $k',k''\in \bigcap_{t=1}^{n-s}J_t^c$. Hence we can pick $n-s$ shifted index sets such that $|\bigcup_{t=1}^{n-s}J_t|\leq n-2$. 
\end{proof}

\bibliographystyle{myIEEEtran}
\bibliography{IEEEabrv,D:/Research/library_organized}

\end{document}